\documentclass{article}
\usepackage{amsmath,amsthm,amssymb}
\usepackage{tikz}
\usetikzlibrary{positioning, arrows.meta, shapes.geometric, backgrounds, fit, calc}
\usepackage{array}
\usepackage[margin=1in]{geometry}
\usepackage{booktabs}
\usepackage{siunitx}
\usepackage{xcolor}
\usepackage{verbatim}
\usepackage{float}
\usepackage{multirow}
\usepackage{graphicx}
\usepackage{placeins}
\usepackage{authblk}
\usepackage{caption}
\usepackage{subcaption}
\usepackage[colorlinks=true, allcolors=blue]{hyperref} 
\usepackage{cleveref}   
\definecolor{revisioncolor}{RGB}{0,120,90}
\newcommand{\rev}[1]{\textcolor{revisioncolor}{#1}}

\newtheorem{proposition}{Proposition}[section]

\theoremstyle{remark}
\newtheorem{remark}{Remark}[section]

\title{Physics-constrained inference of somatic dynamics from dendritic recordings with sparse somatic supervision in weakly coupled two-compartment neuron model}

\author[1]{Abdeltif Oujbara\thanks{Corresponding author. Email: \texttt{abdeltif.oujbara@univ-lehavre.fr}}}
\author[1,2]{Benjamin Ambrosio}
\author[1]{M.\,A. Aziz-Alaoui}

\affil[1]{Department of Mathematics; Normandie Univ.; LMAH UR 3821, Universit\'e Le Havre Normandie, 25 rue Philippe Lebon, 76600 Le Havre, Normandie, France}
\affil[2]{The Hudson School of Mathematics, 244 Fifth Avenue, New York, NY 10001, USA}

\begin{document}

\maketitle
\begin{abstract}

Somatic membrane potential is the primary determinant of neuronal output, yet it remains inaccessible in many experimental setups where only dendritic recordings are available. Reconstructing somatic dynamics from distal measurements is a challenging inverse problem, particularly when the soma and dendrites are weakly coupled, as dendritic signals represent a filtered and attenuated version of somatic activity.
To address this, we use a physics-informed neural network (PINN) constrained by a two-compartment Hodgkin--Huxley model. The network is trained on dense dendritic voltage recordings and the known injected current, complemented by a small number of somatic voltage samples (at most 5\% of the time points), and it reconstructs the full somatic trajectory while adjusting a selected set of somatic maximal conductances.
On synthetic data from four stimulation protocols, we quantify how the reconstruction depends on the amount of somatic supervision. Without somatic samples, the present formulation returns a smooth trajectory in which the action potentials are absent and the subthreshold level is biased, with a root-mean-square error of 10--14~mV; 1\% of the somatic time points is enough to recover every spike; and 5\% brings the root-mean-square error to about 2~mV, spikes included, with relative errors below 0.1\% on the sodium and delayed-rectifier conductances. We also compare the PINN with an unscented Kalman filter constrained by the same model and assimilating the same observations, and we assess robustness to measurement noise and to random initialization.
The reconstructions reported here therefore rely on sparse somatic anchoring in addition to the dendritic recordings. The results delimit what can be inferred in this synthetic, weakly supervised two-compartment setting and identify the amount of somatic information required by the present formulation.

\end{abstract}

\section{Introduction}
The fundamental challenge in computational neuroscience lies in bridging the gap between the biophysical properties of single neurons and the system-level computations that underlie cognition and behavior. Computational models provide a principled way to address this problem: they encode diverse aspects of a system into mathematical equations and then predict its response to prescribed inputs. In neuroscience, such models typically fall into two broad classes. Phenomenological models, including connectionist and statistical approaches (e.g., artificial neural networks)\cite{li2007modeling}, aim to reproduce input--output relationships and abstract computations without explicitly modeling the underlying biology. In contrast, biophysical models seek to represent the mechanisms directly, using conductance based descriptions of ion channels, synapses, and cellular morphology.

Biophysical models can be constructed at multiple levels of organization, from intracellular signaling and gene regulatory pathways to single cell dynamics and large scale networks. At the cellular level, conductance-based models are used to capture firing patterns and the effects of pharmacological manipulation on channel conductances. At the network level, interacting populations of neurons are assembled into circuits that support specific functions for example, the fear circuit involving amygdala, prefrontal cortex, and hippocampus, which is central to emotional regulation, fear acquisition and extinction, and disorders such as PTSD\cite{ledoux2013emotion,rho2023emotional,tovote2015neuronal}. The most powerful strategy is to couple these cellular and network scales, using biologically grounded single-neuron models as building blocks for physiologically realistic circuit simulations.

Within this framework, the Hodgkin--Huxley (HH) formalism \cite{hodgkin1952quantitative} provides a cornerstone for modeling action potential generation via voltage-gated ion channels. However, a critical bottleneck is the determination of model parameters, such as maximal conductances and gating kinetics. Conventional approaches based on voltage clamp and patch clamp are technically demanding, low throughput, and often impractical for systematically characterizing all compartments of morphologically complex neurons.

This challenge is magnified when the objective is to simulate and understand entire neural networks. To study core computational properties such as neuronal synchronization\cite{Maa2024,math12091382}, rhythmic oscillations\cite{aa9ac9aed729454bbdf9588867ace9d1,faber2001laPrincipalNeurons}, and synaptic plasticity (long-term potentiation and depression, LTP/LTD) within circuits such as the amygdala, prefrontal and hippocampal network processes fundamental to emotional regulation, fear extinction, and memory, we must first have accurate and constrained models of the constituent neurons. Realistic multicompartment models are often computationally prohibitive for this network-scale task, leading to the widespread use of reduced models, such as two-compartment Pinsky-Rinzel type models \cite{traub1991model}, which preserve essential dynamics while maintaining computational practicality.

A central obstacle in developing and constraining these models is partial observability. Experimental recordings, especially in vivo and in vitro \cite{aa9ac9aed729454bbdf9588867ace9d1,faber2001laPrincipalNeurons}, are typically limited to one or a few compartments, leaving the states and parameters of hidden compartments unmeasured. This renders the associated inverse problem of parameter estimation ill posed, particularly when the goal is to build a reliable foundation for large scale network studies.

To address this issue, we leverage the emerging framework of physics-informed neural networks (PINNs) \cite{raissi2019}, which seamlessly integrate sparse data with the underlying biophysical laws. PINNs regularize the estimation problem by penalizing the residuals of the governing
differential equations, which can improve state and parameter estimation when the available
data are incomplete.

Recent computational advances have successfully applied this framework to identify parameters in dynamical systems relevant to neuroscience. For instance, \cite{karniadakis2021physics} reviewed the broad potential of physics-informed learning for solving forward and inverse problems in complex biological systems. More specifically, \cite{Rudi2022FHNNNEstimation} demonstrated the efficacy of neural networks in estimating parameters for the FitzHugh-Nagumo model, a simplified reduction of neuronal dynamics. Moving towards biophysical realism, \cite{yao2023bioe} introduced BioE-PINN, a framework tailored for bioelectrical signals, showing robust parameter identification in single-compartment Hodgkin-Huxley models even with limited data. Ferrante et al. \cite{Ferrante2022PINNHH} further validated this approach on real data from the squid giant axon. However, these studies predominantly focus on single point neuron models. The extension of this framework to multicompartment morphologies where spatial coupling and distinct channel distributions play a critical role remains an open challenge that this work aims to address.
In this work, we employ PINNs to solve a critical inverse problem in a two compartment model of a pyramidal neuron. Our contribution is to show that, in this synthetic two-compartment setting, dense dendritic
voltage traces complemented by a small fraction of somatic samples (at most 5\% of the time
points) allow the method to reconstruct the hidden somatic dynamics, and to quantify how this
reconstruction degrades when the somatic samples are removed. Constraining biophysical models
from partially observed data in this way is one step toward building biologically grounded,
large-scale network models. Such models are ultimately essential for simulating and understanding how emergent properties, such as interregional synchronization and learning rules, arise from cellular-level interactions to govern brain function and dysfunction.
\section{Biophysical model formulation}
\label{sec:biophysical_model}

We employ a reduced two-compartment conductance-based model to represent a pyramidal neuron, specifically tailored to capture the electrophysiological properties of cells in the basolateral amygdala (BLA)\cite{ledoux2013emotion,li2009fearNetwork}. This reduction preserves the essential non-linear interaction between the somatic spike-generation zone and the dendritic integration zone, while maintaining a manageable computational complexity for the inverse problem.
\begin{figure}[h!]
    \centering
     \includegraphics[width=0.6\textwidth]{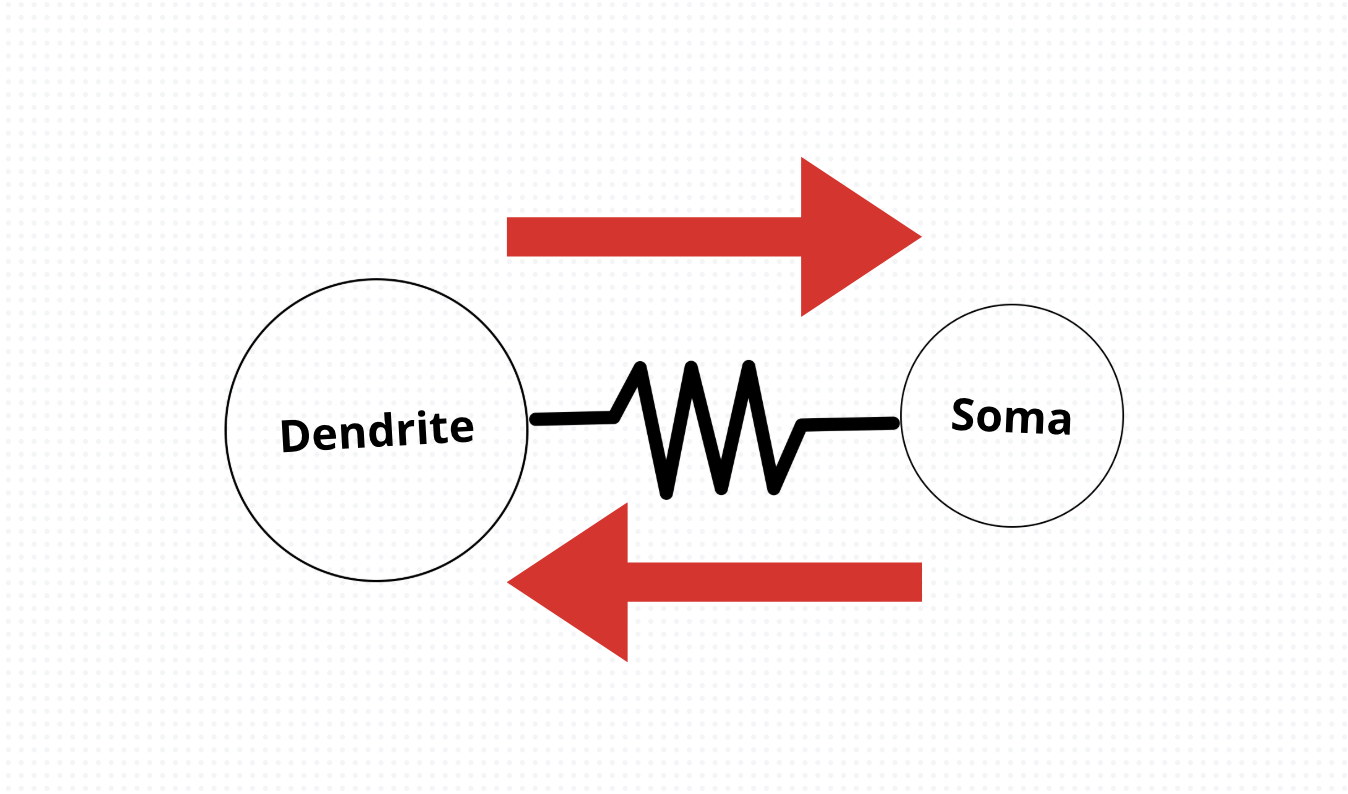} 
     \caption{Two-compartment soma–dendrite neuron model with bidirectional axial coupling $g_c$}
     \label{soma_dendritc_model}
\end{figure}

\subsection{Topology and voltage dynamics}
The neuron is modeled as two electrically coupled compartments  (see Figure~\ref{soma_dendritc_model}) : a soma ($s$) and a dominant apical dendrite ($d$). The membrane potentials $V_s(t)$ and $V_d(t)$ evolve according to the conservation of charge, governed by the following system of coupled ordinary differential equations:

\begin{equation}
\label{eq:voltage_dynamics}
\left\{
\begin{aligned}
    C_{m,s}\frac{dV_s(t)}{dt} &= - \bar{g}_{L,s}\, (V_s(t) - E_L) - I_{\mathrm{ion},s}(V_s(t), \mathbf{x}_s(t)) - \frac{g_c}{p_s}(V_s(t) - V_d(t)) + I_{\mathrm{stim},s}(t), \\
    C_{m,d}\frac{dV_d(t)}{dt} &= - \bar{g}_{L,d}\, (V_d(t) - E_L) - I_{\mathrm{ion},d}(V_d(t), \mathbf{x}_d(t), [\mathrm{Ca}](t)) - \frac{g_c}{p_d}(V_d(t) - V_s(t)) + I_{\mathrm{stim},d}(t).
\end{aligned}
\right.
\end{equation}

Here, \(C_{m,k}\) is the membrane capacitance of compartment \(k\in\{s,d\}\).The first term on the right-hand side of each equation represents the passive leak current, where $\bar{g}_L$ denotes the maximal leak conductance and $E_L$ the leak reversal potential. The term $I_{\mathrm{ion},k}$ groups the active (voltage- and calcium-dependent) ionic currents of compartment $k$, whose gating variables are collected in the state vector $\mathbf{x}_k$; the detailed expressions of these currents are given in Section~\ref{sec:ionic_currents}. The axial coupling between the two compartments is governed by the conductance $g_c$, scaled by the surface-area ratios $p_s$ and $p_d$ (normalized such that $p_s + p_d = 1$). Finally, $I_{\mathrm{stim},k}(t)$ represents the external current injection applied to compartment $k$ during the experimental protocols.

\subsection{Ionic currents}\label{sec:ionic_currents}
The total ionic current $I_{\mathrm{ion},k}$ in compartment $k \in \{s, d\}$ is the sum of specific voltage-gated and calcium-dependent currents. Each current $I_{X,k}$ is modeled using the Hodgkin-Huxley formalism:
\begin{equation}
I_{X,k} \;=\; \bar g_{X,k}\, \Pi_X(\mathbf{x}_k)\, \big(V_k - E_X\big),
\qquad X\in\{Na,DR,M,H,D,Ca,C,sAHP\},\; k\in\{s,d\},
\end{equation}
where \(\bar g_{X,k}\) is the maximal conductance, \(E_X\) the Nernst reversal potential,and \(\Pi_X(\mathbf{x}_k)\) a product of gating variables in compartment \(k\). The model includes the following currents(see Table~\ref{tab:currents} for the complete formulation):
\begin{itemize}
    \item \textbf{Soma:} fast sodium ($I_{Na}$), delayed rectifier potassium ($I_{DR}$), high-threshold calcium ($I_{Ca}$), and M-type potassium ($I_M$).
    \item \textbf{Dendrite:} In addition to $I_{Na}$, $I_{DR}$, $I_{Ca}$, and $I_M$, the dendrite includes hyperpolarization-activated ($I_H$), D-type potassium ($I_D$), Fast calcium-dependent potassium ($I_C$), and slow AHP ($I_{sAHP}$).
\end{itemize}
The complete set of maximal conductances and passive parameters is detailed in Table \ref{tab:parameters}.
\begin{table}[h!]
\centering
\renewcommand{\arraystretch}{1.3}
\caption{Formulation of ionic currents in the soma ($s$) and dendrite ($d$). 
Here $k \in \{s,d\}$ denotes the compartment.}
\label{tab:currents}
\begin{tabular}{ll|ll}
\hline
\multicolumn{2}{c|}{\textbf{Somatic \& dendritic currents}} 
  & \multicolumn{2}{c}{\textbf{Dendritic-specific currents}} \\
\hline
$I_{Na,k}$  
  & $= \bar{g}_{Na,k}\, m_k^3 h_k \,(V_k - E_{Na})$
  & $I_{H,d}$    
  & $= \bar{g}_{H,d}\, r_d \,(V_d - E_H)$ \\

$I_{DR,k}$  
  & $= \bar{g}_{DR,k}\, n_k^4 \,(V_k - E_K)$
  & $I_{D,d}$    
  & $= \bar{g}_{D,d}\, a_d b_d \,(V_d - E_K)$ \\

$I_{Ca,k}$  
  & $= \bar{g}_{Ca,k}\, s_k^2 u_k \,(V_k - E_{Ca})$
  & $I_{C,d}$    
  & $= \bar{g}_{C,d}\, c_d \,(V_d - E_K)$ \\

$I_{M,k}$   
  & $= \bar{g}_{M,k}\, p_k^2 \,(V_k - E_K)$
  & $I_{sAHP,d}$ 
  & $= \bar{g}_{sAHP,d}\, q_d \,(V_d - E_K)$ \\
  &             
  & \\ 
\hline
\end{tabular}
\end{table}
\begin{table}[H]
\centering
\small
\caption{Physiological parameters and maximal conductances for the Two Compartment Model. Values are typically based on Li et al. (2009)\cite{li2009fearNetwork}.}
\label{tab:parameters}
\begin{tabular}{llcc}
\toprule
\multicolumn{4}{c}{\textbf{Passive Properties \& Constants}} \\
\midrule
Parameter & Symbol & Value & Unit \\
\midrule
Membrane Capacitance & $C_m$ & 1 & $\mu\mathrm{F/cm^2}$ \\
Axial coupling conductance & $g_c$ & 0.1 & $\mathrm{mS/cm^2}$ \\
Leak Conductance & $g_L$ & $0.034$ & $\mathrm{mS/cm^2}$ \\
Leak Reversal Potential & $E_L$ & -75 & mV \\
Sodium Reversal Potential & $E_{Na}$ & 45 & mV \\
Potassium Reversal Potential & $E_K$ & -80 & mV \\
Calcium Reversal Potential & $E_{Ca}$ & 120 & mV \\
$I_H$ reversal potential & $E_H$ & -43 & mV \\ 
\midrule
\multicolumn{4}{c}{\textbf{Maximal Conductances ($\bar{g}$) [$\mathrm{mS/cm^2}$]}} \\
\midrule
Channel & Symbol & Soma & Dendrite \\
\midrule
Sodium ($Na$) & $\bar{g}_{Na}$ & 120 & 40 \\
Delayed Rectifier ($DR$) & $\bar{g}_{DR}$ & 12 & 3 \\
M-type Potassium ($M$) & $\bar{g}_M$ & 0.25 & 0.25 \\
High-thresh. Calcium ($Ca$) & $\bar{g}_{Ca}$ & 0.1 & 0.2 \\
Hyperpolarization ($H$) & $\bar{g}_H$ & - & 0.1 \\
D-type Potassium ($D$) & $\bar{g}_D$ & - & 1.0 \\
Fast Ca-dependent K ($C$) & $\bar{g}_C$ & - & 0.5 \\
Slow AHP ($sAHP$) & $\bar{g}_{sAHP}$ & - & 0.1 \\
\bottomrule
\end{tabular}
\end{table}
\subsection{Gating kinetics}\label{Gating kinetics}
Gating variables follow standard first-order kinetics driven by the local membrane voltage $V_k$:
\begin{equation}
\frac{dx}{dt} = \frac{x_\infty(V_k)-x}{\tau_x(V_k)}, 
\qquad x\in\mathcal{X}_k,\;\; k\in\{s,d\},
\label{eq:gates}
\end{equation}
where the gate sets for the somatic and dendritic compartments are defined as:
\[
\mathcal{X}_s=\{m_s,h_s,n_s,p_s,s_s,u_s\}, 
\qquad 
\mathcal{X}_d=\{m_d,h_d,n_d,p_d,r_d,a_d,b_d,s_d,u_d,c_d,q_d\}.
\]
The steady-state activation $x_\infty(V_k)$ and the time constant $\tau_x(V_k)$ are derived from experimental fits\cite{li2009fearNetwork,li2007modeling}. Each gating variable in $\mathcal{X}_s$ and $\mathcal{X}_d$ corresponds to a specific ionic channel gate, following the classical Hodgkin--Huxley naming conventions (see Table~\ref{tab:kinetics} for the complete formulation):
\begin{itemize}
    \item $m$ and $h$: activation and inactivation gates of the fast sodium current $I_{\mathrm{Na}}$ ;
    \item $n$: activation gate of the delayed-rectifier potassium current $I_{\mathrm{DR}}$;
    \item $p$: activation gate of the M-type potassium current $I_{\mathrm{M}}$;
    \item $s$ and $u$: activation and inactivation gates of the high-threshold calcium current $I_{\mathrm{Ca}}$;
    \item $r$: activation gate of the hyperpolarization-activated current $I_{\mathrm{H}}$ ;
    \item $a$ and $b$: activation and inactivation gates of the D-type potassium current $I_{\mathrm{D}}$;
    \item $c$: activation gate of the fast calcium-dependent potassium current $I_{\mathrm{C}}$;
    \item $q$: activation gate of the slow afterhyperpolarization current $I_{\mathrm{sAHP}}$.
\end{itemize}

The somatic set $\mathcal{X}_s = \{m_s, h_s, n_s, p_s, s_s, u_s\}$ thus contains six gating variables corresponding to the four somatic currents ($I_{\mathrm{Na}}$, $I_{\mathrm{DR}}$, $I_{\mathrm{M}}$, $I_{\mathrm{Ca}}$), while the dendritic set $\mathcal{X}_d$ contains eleven variables, reflecting the four additional dendrite-specific currents ($I_{\mathrm{H}}$, $I_{\mathrm{D}}$, $I_{\mathrm{C}}$, $I_{\mathrm{sAHP}}$). The power exponents (e.g.\ $m^3$, $n^4$, $s^2$) indicate the number of independent and identical subunits that must simultaneously be in the open state for the channel to conduct, as originally proposed by Hodgkin and Huxley~\cite{hodgkin1952quantitative}.

These functions can be equivalently expressed through an $\alpha$--$\beta$ parametrization, corresponding to the two-state Markov scheme shown in Figure \ref{fig:gating_scheme}.

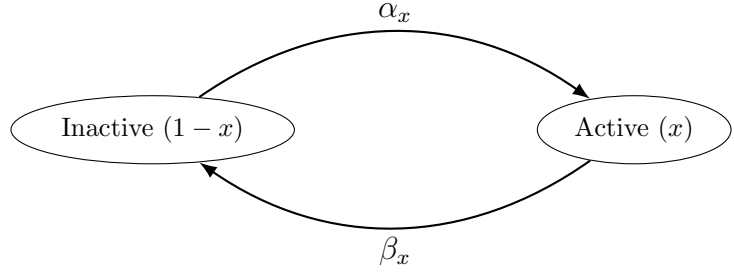
\begin{figure}[h]
\centering
\begin{tikzpicture}[
    >=Latex,
    node distance=3.2cm,
    state/.style={draw, ellipse, minimum width=2.2cm, minimum height=0.9cm},
    lab/.style={font=\large}
]
\node[state] (I) {Inactive ($1-x$)};
\node[state, right=of I] (A) {Active ($x$)};
\draw[->, thick, bend left=35]  (I) to node[above, lab] {$\alpha_x$} (A);
\draw[->, thick, bend left=35]  (A) to node[below, lab] {$\beta_x$} (I);
\end{tikzpicture}
\caption{Two-state gating scheme (Inactive/Active).}
\label{fig:gating_scheme}
\end{figure}
In this formulation, the dynamics follow:
\begin{equation}
\frac{dx}{dt} = \alpha_x(V_k)\,(1-x) - \beta_x(V_k)\,x,
\label{eq:alpha_beta}
\end{equation}
where the conversion between the two formalisms is given by:
\begin{equation}
x_\infty(V_k)=\frac{\alpha_x(V_k)}{\alpha_x(V_k)+\beta_x(V_k)},
\qquad
\tau_x(V_k)=\frac{1}{\alpha_x(V_k)+\beta_x(V_k)}.
\label{eq:xinf_tau_from_alpha_beta}
\end{equation}

Unless stated otherwise, we assume identical kinetic rate functions in the soma and dendrite. However, distinct gating states are retained ($x_s$ vs $x_d$) because the local membrane potentials $V_s(t)$ and $V_d(t)$ generally evolve differently.

\begin{table}[H]
\centering
\footnotesize
\caption{Kinetic functions for gating variables. Note: For $I_H$, $I_D$, $I_{Ca}$, and $I_{sAHP}$, the formulations are given directly as steady-state ($x_\infty$) and time constant ($\tau_x$).}
\label{tab:kinetics}
\setlength{\tabcolsep}{4pt}
\renewcommand{\arraystretch}{1.8} 
\begin{tabular}{llcc}
\toprule
\textbf{Current} & \textbf{Gate} & \textbf{Forward Rate $\alpha(V)$} & \textbf{Backward Rate $\beta(V)$ } \\
\midrule
\multirow{2}{*}{$I_{Na}$} 
 & $m$ & $\alpha_m = \dfrac{-0.2816(V+25)}{e^{-(V+25)/9.3}-1}$ 
     & $\beta_m = \dfrac{0.2464(V-2)}{e^{(V-2)/6}-1}$ \\
 & $h$ & $\alpha_h = 0.098 \, e^{-(V+40.1)/20}$ 
     & $\beta_h = \dfrac{1.4}{1 + e^{-(V+10.1)/10}}$ \\
\midrule
$I_{DR}$ 
 & $n$ & $\alpha_n = \dfrac{-0.036(V-13)}{e^{-(V-13)/25}-1}$ 
     & $\beta_n = \dfrac{0.0108(V-23)}{e^{(V-23)/12}-1}$ \\
\midrule
$I_{M}$ 
 & $p$ & $\alpha_p = \dfrac{0.016}{1 + e^{-(V+52.7)/23}}$ 
     & $\beta_p = \dfrac{0.016}{1 + e^{(V-52.7)/18.8}}$ \\
\midrule
$I_{C}$ 
 & $c$ & $\alpha_c = \dfrac{-0.00642 V_m - 0.1152}{e^{-(V_m+18)/12}-1}$ 
     & $\beta_c = 1.7 \, e^{-(V_m+152)/30}$ \\
 & & \multicolumn{2}{l}{\textit{where} $V_m = V + 40\log_{10}([\mathrm{Ca}]_1)$} \\
\midrule
\midrule
\textbf{Current} & \textbf{Gate} & \textbf{Steady State $x_\infty(V)$} & \textbf{Time Constant $\tau_x(V)$ (ms)} \\
\midrule
$I_{H}$ & $r$ 
    & $\dfrac{1}{1 + e^{(V+89.2)/9.5}}$ 
    & $1727 \, e^{0.019 V}$ \\

\midrule
\multirow{2}{*}{$I_D$}
 & $a$ 
    & $\dfrac{1}{1 + e^{-(V+8.6)/11.1}}$ 
    & $1.5$ \\
 & $b$ 
    & $\dfrac{1}{1 + e^{(V+35)/10}}$ 
    & $569$ \\    
   
\midrule
\multirow{2}{*}{$I_{Ca}$} 
 & $s$ & $\dfrac{1}{1 + e^{(V+21)/9}}$ & $569$ \\
 & $u$ & $\dfrac{1}{1 + e^{(V+24.6)/11.3}}$ & $1.25 \, \mathrm{sech}(-0.031(V+37.1))$ \\
\midrule
$I_{sAHP}$ & $q$ 
    & $\dfrac{0.0048}{1 + e^{-5\log_{10}([\mathrm{Ca}]_2)-17.5}}$ 
    & $\tau_q = 48$ \\
\bottomrule
\end{tabular}
\end{table}
\subsection{Calcium dynamics}

To account for the calcium-dependent currents used in this model, we track two intracellular calcium pools, $[\mathrm{Ca}^{2+}]_1$ and $[\mathrm{Ca}^{2+}]_2$, following a two-pool submembrane-shell approximation \cite{warman1994reconstruction,mainen1998activeDendrites}. The dynamics of these pools are governed by first-order differential equations of the form\cite{warman1994reconstruction,li2008regulationITC,durstewitz2000dopamine}:

\begin{equation}
\frac{d[\mathrm{Ca}^{2+}]_i}{dt} = - \frac{f_i \, I_{Ca,d}}{z F V} + \frac{[\mathrm{Ca}^{2+}]_{\mathrm{rest}} - [\mathrm{Ca}^{2+}]_i}{\tau_i}, \quad i \in \{1, 2\}.
\label{eq:ca_dynamics}
\end{equation}

Here,$f_i$ is the fraction of the total calcium influx captured by pool $i$, $z=2$ is the valence of the $\mathrm{Ca}^{2+}$ ion, $F$ is the Faraday constant, and $V$ is the volume of a dendritic shell with thickness $w = 1\ \mu\text{m}$. The resting and initial concentration is set to $[\mathrm{Ca}^{2+}]_{\mathrm{rest}} = 50\ \text{nmol/l}$. The model distinguishes between two functional calcium pathways, whose parameters were established in prior experimental and computational studies~\cite{warman1994reconstruction,li2009fearNetwork,durstewitz2000dopamine}:
\begin{itemize}
    \item \textbf{Pool 1 (Fast):} Mediates the rapid activation of $I_C$. The influx fraction $f_1 = 0.7$ reflects the proximity of the fast calcium-dependent potassium channels to the calcium entry sites, ensuring a large and rapid rise in local $[\mathrm{Ca}^{2+}]$. The short decay time constant $\tau_1 = 1$\,ms captures the fast buffering and diffusion that limit the duration of this transient~\cite{warman1994reconstruction}.
    \item \textbf{Pool 2 (Slow):} Drives the slow adaptation current $I_{\mathrm{sAHP}}$. The smaller influx fraction $f_2 = 0.024$ accounts for the greater distance of the sAHP channels from the calcium sources, resulting in a much smaller effective calcium signal. The long removal time constant $\tau_2$ reflects the slow clearance mechanisms (pumps, exchangers, and cytoplasmic buffering) governing this distal pool~\cite{durstewitz2000dopamine,li2009fearNetwork}.
\end{itemize}

\section{Impact of coupling strength on somato-dendritic synchronization}
\label{subsec:coupling_dynamics}

To quantify the role of the axial coupling conductance $g_c$ in shaping the joint soma--dendrite dynamics, we performed forward simulations of the full two-compartment model across a range of coupling strengths.
As established in seminal work on reduced neuronal models \cite{pinsky1994intrinsic, mainen1996influence}, $g_c$ controls the rate of axial charge transfer between soma and dendrite and thus determines the degree of electrical compartmentalization.
In all simulations, stimulation was applied to the dendrite only ($I_{\mathrm{stim},s}(t)=0$), and the system was initialized with distinct somatic and dendritic states (e.g., $V_s(0)\neq V_d(0)$; gating variables set to steady state at the corresponding initial voltages unless stated otherwise).
As illustrated in Figure~\ref{fig:coupling_synchronization}, three qualitatively distinct dynamical regimes emerge, in line with asymptotic analyses of compartmental systems \cite{keener2009mathematical}:

\begin{figure}[h!]
    \centering
    \includegraphics[width=\textwidth]{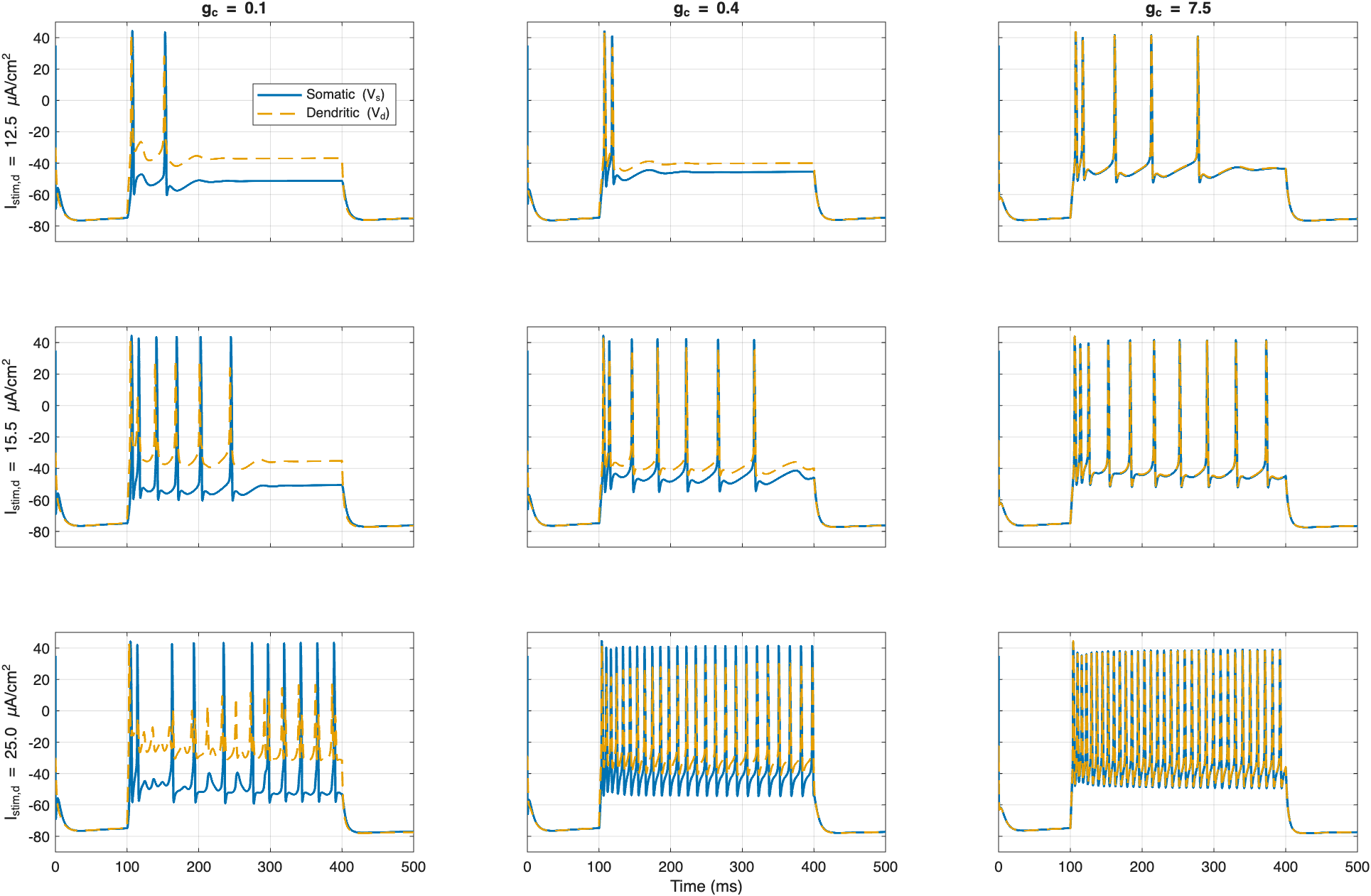}
    \caption{\textbf{Coupling-dependent somato--dendritic synchronization and firing under dendritic stimulation.}
Columns: coupling conductance $g_c\in\{0.1,0.4,7.5\}$; rows: dendritic injected current $I_{\mathrm{stim},d}\in\{12.5,15.5,25.0\}\,\mu\mathrm{A/cm^2}$.
Solid blue: somatic voltage $V_s$; dashed orange: dendritic voltage $V_d$.
Stimulation is applied to the dendrite only during $t\in[100,400]$~ms ($I_{\mathrm{stim},s}=0$), and voltage dynamics follow Eq.~\eqref{eq:voltage_dynamics}.
Simulations are initialized with different somatic and dendritic initial conditions ($V_s(0)\neq V_d(0)$).
Increasing $g_c$ reduces the soma-dendrite voltage mismatch and promotes spike-time locking; for a fixed $I_{\mathrm{stim},d}$, stronger coupling also tends to increase the somatic firing rate by improving transmission of dendritic drive to the soma.}
\label{fig:coupling_synchronization}
\end{figure}

\begin{itemize}
    \item \textbf{Weak coupling regime ($g_c \to 0$).} 
    The soma and dendrite behave almost as independent oscillators. Somatic spikes fail to invade the dendrite, which exhibits only small, subthreshold deflections. In this limit, the dendritic compartment behaves as an electrotonically distant structure, only weakly influenced by somatic activity.

    \item \textbf{Intermediate coupling and bidirectional propagation.}
    As $g_c$ increases from near zero to physiologically plausible values (typically in the range $0.3$--$2$\,mS/cm$^2$ for pyramidal neurons 
of the basolateral amygdala~\cite{li2009fearNetwork,pinsky1994intrinsic}), the coupling becomes effectively bidirectional. Large somatic depolarizations can overcome the axial resistance and propagate into the dendritic arbor, while dendritic events can in turn influence somatic excitability. In particular, this regime supports the retrograde propagation of action potentials from soma to dendrite, a phenomenon known as the backpropagating action potential (bAP) \cite{stuart1994active}. 

    \item \textbf{Strong coupling regime ($g_c \to \infty$).}
    In the asymptotic limit of large $g_c$, the axial current dominates the mismatch dynamics and forces $V_d(t)$ to closely track $V_s(t)$ with little delay or attenuation. The two compartments then behave as an approximately isopotential  unit, the residual mismatch
    decreasing as $\mathcal{O}(1/g_c)$ (see Proposition~\ref{prop:sync_bound} 
in Section~\ref{sec:sync_stability}).
\end{itemize}
\subsection{Frequency modulation via dendritic loading}

The coupling conductance $g_c$ influences both soma--dendrite coordination and the input--output gain of the neuron. As $g_c$ varies, we observe two robust effects.
\begin{itemize}
    \item \textbf{Weak coupling produces dendritic attenuation.}
    When $g_c$ is small ($g_c \approx 0.1\,\mathrm{mS/cm^2}$), the axial pathway strongly filters fast somatic events. As a result, dendritic depolarizations remain low-amplitude and subthreshold, whereas the soma continues to generate full action potentials. This behavior is consistent with the electrotonic attenuation predicted by passive cable theory \cite{rall1962electrophysiology}.

\item \textbf{Dendritic loading modulates the somatic $f$--$I$ relationship.}
    For somatic drive, the dendritic compartment acts as an additional capacitive and
    resistive load: it draws axial current during the pre-spike depolarization, reduces the
    somatic depolarization slope and lowers the firing rate, an effect documented in
    pyramidal neurons \cite{eyal2014dendritic}. Under the dendritic stimulation used in
    Figure~\ref{fig:coupling_synchronization} the balance is reversed, since the injected
    current must first cross the axial pathway to reach the spike-generating zone: increasing
    $g_c$ then improves the transmission of the dendritic drive and increases the somatic
    firing rate. In both cases $g_c$ acts as a biophysical gain-control parameter, and its
    effect on the firing rate depends on where the current is injected.
 
\end{itemize}

The coupling conductance therefore does more than shape synchronization: it also sets the gain
between the injected current and the somatic firing rate, through the resistive and capacitive
load that each compartment imposes on the other.

\subsection{Asymptotic bound on the soma–dendrite voltage mismatch}\label{sec:sync_stability}

To formally analyze the stability of the synchronized state, we define the synchronization error
\[
    e(t) = V_d(t) - V_s(t).
\]
Let $S(V_s, \mathbf{x}_s)$ and $D(V_d, \mathbf{x}_d, [\mathrm{Ca}]_d)$ denote the sum of intrinsic ionic currents in the soma and dendrite, respectively, where $\mathbf{x}_k$ represents all gating variables in compartment $k$:
\begin{align*}
S(V_s, \mathbf{x}_s) &=  -\Big(I_{Na,s} + I_{DR,s} + I_{M,s} + I_{Ca,s} + I_{L,s}\Big), \\
D(V_d, \mathbf{x}_d, [\mathrm{Ca}]_d) &= -\Big(I_{Na,d} + I_{DR,d} + I_{M,d} + I_{H,d} + I_{D,d} + I_{Ca,d} + I_{C,d} + I_{sAHP,d} + I_{L,d}\Big).
\end{align*}

Assuming the same specific membrane capacitance $C_m$ in both compartments and subtracting the
somatic charge-balance equation from the dendritic one in \eqref{eq:voltage_dynamics} yields the
error dynamics
\begin{equation}
    C_m \dot{e}(t) = D(V_d, \mathbf{x}_d, [\mathrm{Ca}]_d) - S(V_s, \mathbf{x}_s)
    - \gamma\, e(t) + I_{\mathrm{stim},d} - I_{\mathrm{stim},s},
    \qquad
    \gamma = g_c\left(\frac{1}{p_s}+\frac{1}{p_d}\right),
    \label{eq:error_dynamics_full}
\end{equation}
where $\gamma$ is the effective coupling coefficient. With the surface-area normalization
$p_s+p_d=1$ used throughout, and $p_s=p_d=1/2$, this gives $\gamma=4g_c$.

\paragraph{A dissipative bound on the frozen-gate voltage term.}
Let $\Omega=[V_{\min},V_{\max}]\times[0,1]^{N_x}\times[0,\mathrm{Ca}_{\max}]$ be a physically
relevant bounded invariant set for the full state $(V,\mathbf x,[\mathrm{Ca}])$, where $N_x$ is
the total number of gating variables. All Hodgkin--Huxley currents and gating kinetics are
continuously differentiable, so $S$ and $D$ are $C^1$ on $\Omega$ and their Jacobians are
bounded by compactness. We split the ionic difference as
\begin{equation}
D(V_d, \mathbf{x}_d, [\mathrm{Ca}]_d) - S(V_s, \mathbf{x}_s)
= \underbrace{D(V_d, \mathbf{x}_d, [\mathrm{Ca}]_d) - D(V_s, \mathbf{x}_d, [\mathrm{Ca}]_d)}_{\text{voltage term}}
+ \underbrace{D(V_s, \mathbf{x}_d, [\mathrm{Ca}]_d) - S(V_s, \mathbf{x}_s)}_{\Delta_{\mathrm{state}}(t)} .
\label{eq:decomposition}
\end{equation}
For fixed gating variables and calcium concentrations, every current is ohmic in $V$, so the
voltage term is governed by the frozen-gate slope conductance
\[
\frac{\partial D}{\partial V}(V,\mathbf x,[\mathrm{Ca}])
=-\Big(\bar g_{L,d}+\sum_{X}\bar g_{X,d}\,\Pi_X(\mathbf x)\Big)
\;\le\;-\bar g_{L,d}\;=:\;-g_{\min,d}\;<\;0,
\]
because $\Pi_X(\mathbf x)\ge 0$ for every current $X$. The mean value theorem then gives, for any
$V_s,V_d$,
\begin{equation}
\operatorname{sgn}(e)\big[D(V_d,\mathbf x_d,[\mathrm{Ca}]_d)-D(V_s,\mathbf x_d,[\mathrm{Ca}]_d)\big]
\;\le\;-g_{\min,d}\,|e| .
\label{eq:dissipative_voltage_term}
\end{equation}
The voltage term is therefore dissipative: it contracts the mismatch instead of amplifying it.
The second term, $\Delta_{\mathrm{state}}(t)$, captures the differences in gating variables and
calcium concentrations between the two compartments, as well as the difference in channel
complement (the dendrite-specific currents). It is uniformly bounded on $\Omega$, and together
with the stimulation difference we set
\[
\Delta_{\mathrm{total}}(t)=\Delta_{\mathrm{state}}(t)+I_{\mathrm{stim},d}(t)-I_{\mathrm{stim},s}(t),
\qquad |\Delta_{\mathrm{total}}(t)|\le B_{\mathrm{tot}} .
\]
 
\begin{proposition}[Practical synchronization]
\label{prop:sync_bound}
Let $\gamma=g_c\,(1/p_s+1/p_d)$ be the effective coupling coefficient,
$g_{\min,d}=\bar g_{L,d}>0$ the minimal frozen-gate slope conductance of the dendritic
compartment, and $B_{\mathrm{tot}}$ the uniform bound on $\Delta_{\mathrm{total}}(t)$. Then, for
every $g_c>0$, the voltage mismatch $e(t)=V_d(t)-V_s(t)$ satisfies
\[
|e(t)|\;\le\;|e(0)|\,e^{-(\gamma+g_{\min,d})\,t/C_m}
+\frac{B_{\mathrm{tot}}}{\gamma+g_{\min,d}}
\Big(1-e^{-(\gamma+g_{\min,d})\,t/C_m}\Big),
\]
and in particular
\[
\limsup_{t\to\infty}|V_d(t)-V_s(t)|\;\le\;\frac{B_{\mathrm{tot}}}{\gamma+g_{\min,d}}
\;=\;\mathcal O\!\left(\frac{1}{g_c}\right)\quad\text{as } g_c\to\infty .
\]
The mismatch therefore converges exponentially to a ball whose radius shrinks like $1/g_c$; it
converges to zero only in the degenerate case $B_{\mathrm{tot}}=0$.
\end{proposition}
 
\begin{proof}
Multiplying \eqref{eq:error_dynamics_full} by $\operatorname{sgn}(e)$ and using
\eqref{eq:dissipative_voltage_term} together with $|\Delta_{\mathrm{total}}|\le B_{\mathrm{tot}}$
gives
\[
C_m\frac{d|e|}{dt}\;=\;C_m\operatorname{sgn}(e)\,\dot e
\;\le\;-(\gamma+g_{\min,d})\,|e|+B_{\mathrm{tot}} .
\]
Since $\gamma+g_{\min,d}>0$, a Gr\"onwall-type argument yields the stated estimate.
\end{proof}
 
\begin{remark}
The estimate is a statement of practical synchronization: it bounds the asymptotic voltage
mismatch and does not assert asymptotic stability of an exact synchronized state. The two
compartments carry different channel complements, so $B_{\mathrm{tot}}>0$ and a residual
mismatch persists at any finite coupling. No lower bound on $g_c$ is required: the contraction
of the voltage mismatch is already provided by the dendritic leak, and the coupling only
accelerates it and shrinks the residual ball. This is sharper than the classical worst-case
argument, in which the frozen-gate term is bounded in absolute value by
$L_{\mathrm{volt}}=\bar g_{L,d}+\sum_X\bar g_{X,d}$ and a condition of the form
$\gamma>L_{\mathrm{volt}}$ is imposed: with the dendritic conductances of
Table~\ref{tab:parameters}, $L_{\mathrm{volt}}\approx45\ \mathrm{mS/cm^2}$, so that condition
would require $g_c\gtrsim11\ \mathrm{mS/cm^2}$, above every coupling simulated here, whereas
Figure~\ref{fig:coupling_synchronization} already shows near-isopotential behaviour at
$g_c=7.5\ \mathrm{mS/cm^2}$. Finally, the bound concerns the voltage mismatch only; the gating
and calcium mismatches are absorbed in $B_{\mathrm{tot}}$, which is why the residual error does
not vanish. The $\mathcal O(1/g_c)$ decay is consistent with theoretical results on electrically
coupled neuronal oscillators \cite{chow2000dynamics} and, qualitatively, with the coupled
bursting systems analysed by Belykh et al.~\cite{belykh2005synchronization}. The reduction to a
quasi-isopotential regime at large $g_c$ justifies, a posteriori, the use of single-compartment
approximations in the high-conductance limit \cite{keener2009mathematical}.
\end{remark}

\subsection{The inverse problem in the weak coupling regime}

The analytical results and numerical simulations above demonstrate that weak coupling ($g_c \lesssim 0.5\,\mathrm{mS/cm^2}$) leads to a pronounced decoupling of somatic and dendritic dynamics: the dendritic trace $V_d(t)$ becomes a strongly low-pass filtered and attenuated version of the somatic output.

This desynchronization is precisely the central challenge addressed in this work. In this regime, simple correlation measures or linear extrapolation are fundamentally insufficient to reconstruct the rich, nonlinear somatic spiking dynamics from the smooth dendritic recordings; the information appears locally lost.

However, our approach relies on the hypothesis that the information is not destroyed, but rather transformed by the coupling equation
\begin{equation}
    I_{\mathrm{axial}}(t) = g_c \,\bigl(V_s(t) - V_d(t)\bigr).
\end{equation}
Embedding this constraint, together with the two current-balance equations, in the loss
function of a neural network is a way of testing that hypothesis. We therefore ask, in the
weak-to-moderate coupling regime, how much of the somatic spiking activity can be recovered
from the dendritic trace, and how much additional somatic information is needed. The answer
given in Section~\ref{sec:results} is specific to the synthetic two-compartment setting, the
four stimulation protocols and the coupling value considered here: within this setting, the
dendritic observations alone were not sufficient under the present loss, and a small fraction
of somatic samples had to be added.

\subsection{Relationship to classical and PINN-based inverse methods}
\label{sec:classical_pinn_inverse_methods}

Inverse problems in conductance-based neuronal models have a long history in computational neuroscience. In their most general form, these problems consist in estimating hidden states, unknown parameters, or sometimes unknown inputs from partial electrophysiological observations. Depending on the experimental setting, the unknown quantities may include maximal conductances, initial conditions, gating variables, kinetic parameters entering the voltage-dependent rate functions, or the unobserved membrane potential of a specific compartment. These problems are difficult because conductance-based models are nonlinear, stiff, partially observed, and often non-identifiable: different combinations of parameters may produce similar voltage traces. Similar difficulties arise in inverse problems for differential equations more generally, where partial or indirect observations often lead to ill-posed reconstruction tasks that require additional regularization or structural constraints~\cite{huntul2022inverse}.

A first class of approaches relies on direct biophysical parameter fitting from electrophysiological protocols. In classical conductance-based modeling, voltage-clamp data are often used to determine ionic current amplitudes and voltage-dependent gating kinetics, while current-clamp recordings are used to tune maximal conductances so that the model reproduces firing patterns. This strategy originates from classical Hodgkin--Huxley-type modeling~\cite{hodgkin1952quantitative}
and remains standard in more detailed conductance-based neuron models,
such as the CA3 pyramidal-cell model of Traub et al.~\cite{traub1991model},
which incorporated voltage-clamp data on intrinsic conductances. However, such procedures require rich experimental protocols and often assume that the relevant compartment is directly recorded. They become much harder when only one compartment is observed and the target dynamics occur in another, hidden compartment.

A second family of methods is based on statistical and optimization-based parameter estimation. Huys et al.~\cite{huys2006efficient} developed efficient methods for estimating parameters of detailed single-neuron models from electrophysiological data, showing that statistical inference can be used to constrain complex conductance-based models. Daly et al.~\cite{daly2015hodgkin} studied the identifiability of the classical Hodgkin--Huxley model and showed that reparametrization can be necessary because some parameters or parameter combinations are difficult to identify uniquely from voltage traces. These studies emphasize an important point for the present work: estimating all parameters of a Hodgkin--Huxley model at once is generally ill-conditioned. For this reason, we restrict the inverse problem to a selected set of somatic maximal conductances while keeping the voltage-dependent gating kinetics fixed.

Kalman-based filtering methods form another important class of inverse
techniques. These methods treat the neuron model as a dynamical state-space
system and recursively update the estimated states and parameters as new
observations become available. Moye and Diekman~\cite{moye2018data} used data
assimilation methods, including Unscented Kalman Filtering, for neuronal state
and parameter estimation. Lankarany et al.~\cite{lankarany2013joint} applied a
Dual Extended Kalman Filter to estimate states and conductances in
Hodgkin--Huxley-type models. Such filters are attractive because they are
sequential and can be used in online settings. However, when parameters are
estimated by augmenting the state vector, the resulting system can become
high-dimensional and sensitive to initialization, model mismatch, and partial
observability. This last limitation is directly relevant to the setting studied
here, and for this reason an unscented Kalman filter is used in
Section~\ref{sec:ukf} as a baseline against which the PINN reconstruction is
evaluated. Both estimators are then constrained by the same two-compartment
model and assimilate the same observations, so that the comparison isolates the
effect of imposing the model globally rather than sequentially. Related state-estimation work has tracked Hodgkin--Huxley dynamics from partial
recordings~\cite{UllahSchiff2009}, identified nonlinear systems from uncertain and indirect
measurements~\cite{Voss2004}, estimated parameters and predicted membrane voltages in
conductance-based models~\cite{Meliza2014BiolCyb,Taylor2020PCBI}, and estimated synaptic
conductances in the presence of subthreshold nonlinearities~\cite{Vich2017Frontiers}.

Variational data assimilation provides a complementary strategy. Instead of updating the solution sequentially, variational methods optimize the hidden trajectory and parameters over an entire observation window. Kadakia et al.~\cite{kadakia2016nonlinear} used nonlinear statistical data assimilation to estimate hidden states and parameters in conductance-based neuronal models, including multi-variable neuronal systems. This full-window formulation is powerful because it uses information from the complete time interval, but it can be computationally expensive and may require careful regularization when the observed variables provide only weak constraints on the hidden dynamics.

More recently, gradient-based approaches have been used to estimate parameters by differentiating through the numerical simulation of the neuron model. In this direction, backpropagation through time (BPTT) has been applied to unrolled Hodgkin--Huxley dynamics for automatic conductance estimation~\cite{li2025bptt}. Such approaches are natural when the model equations are implemented as differentiable recurrent computations. However, they may still face difficulties related to stiffness, long-time gradient propagation, and the identifiability of weak or slow conductances.

PINNs have been increasingly explored for biological inverse problems. Rudi et al.~\cite{Rudi2022FHNNNEstimation} used dense and convolutional neural networks for parameter estimation in the FitzHugh--Nagumo model, a reduced model of excitable dynamics. Herrero Martin et al.~\cite{HerreroMartin2022EPPINNs} introduced EP-PINNs for characterizing cardiac electrophysiology. Yao et al.~\cite{yao2023bioe} proposed BioE-PINN for bioelectrical signals, showing that physics-informed learning can improve parameter estimation when data are limited. Ferrante et al.~\cite{Ferrante2022PINNHH} applied PINNs to inverse problems in Hodgkin--Huxley neuron models. These works demonstrate that PINNs can combine sparse observations with differential-equation constraints to recover hidden states and parameters in nonlinear electrophysiological systems. Beyond electrophysiology, PINN-based methods have also been used for other nonlinear differential-equation models. For instance, Berkhahn and Ehrhardt~\cite{berkhahn2022pinncovid} used a PINN framework for parameter estimation and simulation in an ODE-based dynamical system, while Ali~\cite{ali2025pinnpbm} applied a PINNs approach to nonlinear population balance models. Related neural-network solvers for time-dependent differential equations and ordinary differential equations have also been developed in~\cite{schneidereit2023adaptive,nam2022error}. These studies support the broader methodological idea that neural approximators can be combined with differential-equation constraints to solve forward or inverse problems from limited data.

The present study focuses on a different inverse setting from most of the approaches described above. Here, the main observed signal is the dendritic voltage $V_d(t)$, together with the known dendritic stimulation current $I_{\mathrm{stim},d}(t)$. 
The somatic voltage $V_s(t)$, which represents the spike-generating output of the neuron, is only weakly observed through sparse somatic samples and is otherwise treated as a hidden trajectory. 
Thus, the objective is not simply to estimate parameters from a directly recorded membrane-potential trace. 
Instead, the goal is to recover the somatic output from dendritic measurements by exploiting the soma--dendrite axial coupling and the current-balance equations of a weakly coupled two-compartment Hodgkin--Huxley model. 
In this regime, the dendritic voltage is a filtered and attenuated representation of somatic activity, which makes the inverse problem particularly sensitive to partial observability and compartmental decoupling.

We deliberately keep the inverse problem restricted to a controlled set of unknowns. 
The injected dendritic current is treated as a known input, and the ionic-current expressions, passive parameters, and voltage-dependent gating kinetics are fixed by the biophysical model specified in Tables~\ref{tab:currents},\ref{tab:parameters},\ref{tab:kinetics}. 
The PINN therefore does not learn the gating rate functions themselves. 
The trainable biophysical parameters are limited to the selected somatic maximal conductances
\[
\theta_g=\{\bar g_{Na,s},\bar g_{DR,s},\bar g_{M,s},\bar g_{Ca,s}\}.
\]
This formulation allows us to focus on a well-defined reconstruction problem: recovering the hidden somatic output $V_s(t)$ from dendritic observations under weak coupling, while estimating only a small set of key somatic conductances.

A PINN-based framework is well suited to this partially observed compartmental setting for several reasons. 
First, the physics-informed loss provides a mechanism for propagating information from the observed dendritic compartment to the hidden somatic compartment through the governing equations, in particular the membrane current-balance equations and the soma--dendrite axial coupling term. 
Second, the continuous-time differentiable representation allows the voltage equations, gating dynamics, calcium dynamics, and axial-coupling residuals to be incorporated directly into the loss function. 
Third, the framework supports joint state and parameter estimation: the neural network reconstructs the latent trajectory, while the selected somatic maximal conductances are optimized as global trainable parameters. 
Finally, because PINNs are mesh-free and do not require the solution to be represented only on a fixed numerical grid, the same principle is naturally compatible with future extensions to more complex neuronal morphologies, where information propagates through a dendritic tree governed by compartmental or cable-type equations.

\section{Methods: PINN for somatic voltage reconstruction and conductance identification}
\label{sec:methods_pinn}

Figure~\ref{fig:pinn_framework}   provides a schematic overview of the 
complete inverse framework. We address the inverse problem of reconstructing the hidden somatic voltage trajectory $V_s(t)$ from dense dendritic recordings  complemented by sparse somatic samples, in a two-compartment conductance-based model. The forward dynamics membrane charge balance, Hodgkin–Huxley ionic currents, gating kinetics, and calcium pools are fully specified in Section~\ref{sec:biophysical_model}.

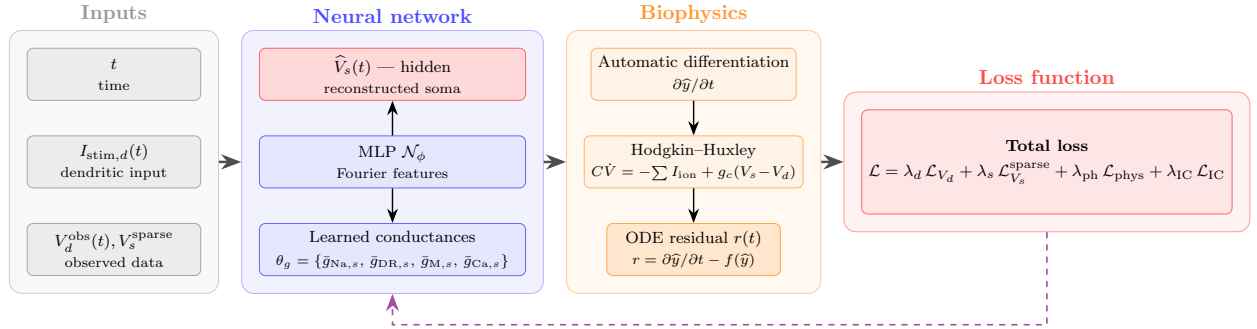
\begin{figure}[H]
\centering
\resizebox{\textwidth}{!}{%
\begin{tikzpicture}[
    font=\small,
    node distance=6mm and 10mm,
    >={Stealth[length=2.5mm]},
    panel/.style={rounded corners=6pt, draw=black!40, line width=0.4pt,
                  inner sep=3mm, align=center},
    block/.style={rounded corners=3pt, draw=black!60, line width=0.4pt,
                  minimum width=30mm, minimum height=9mm, align=center,
                  font=\footnotesize},
    highlight/.style={block, fill=red!15, draw=red!60},
    inputblock/.style={block, fill=gray!15, draw=gray!70},
    physblock/.style={block, fill=orange!10, draw=orange!60},
    nnblock/.style={block, fill=blue!10, draw=blue!60,
                    minimum width=46mm},
    lossblock/.style={block, fill=red!10, draw=red!60,
                      minimum width=46mm, minimum height=18mm},
    arr/.style={-{Stealth[length=2.5mm]}, line width=0.7pt},
    bigarr/.style={-{Stealth[length=3.5mm]}, line width=1.2pt, draw=black!70},
    backprop/.style={-{Stealth[length=3mm]}, line width=0.8pt,
                     dashed, draw=violet!70}
]

\node[inputblock] (t) {$t$\\[1pt]\scriptsize time};
\node[inputblock, below=of t] (istim)
    {$I_{\mathrm{stim},d}(t)$\\[1pt]
     \scriptsize dendritic input};
\node[inputblock, below=of istim] (vd)
    {$V_d^{\mathrm{obs}}(t), V_s^{\mathrm{sparse}}$\\[1pt]
     \scriptsize observed data};

\begin{pgfonlayer}{background}
  \node[panel, fill=gray!5, draw=gray!40,
        fit=(t)(istim)(vd),
        label={[gray!80,font=\bfseries]above:Inputs}]
        (P0) {};
\end{pgfonlayer}

\node[nnblock, right=10mm of istim] (mlp)
    {MLP $\mathcal{N}_\phi$\\[1pt]
     \scriptsize Fourier features};

\node[highlight, above=of mlp, minimum width=46mm] (vs)
    {$\widehat{V}_s(t)$ --- hidden\\[1pt]
     \scriptsize reconstructed soma};
\node[nnblock, below=of mlp] (aux)
    {Learned conductances\\[2pt]
     \scriptsize $\theta_g=\{\bar g_{\mathrm{Na},s},\,\bar g_{\mathrm{DR},s},\,
                            \bar g_{\mathrm{M},s},\,\bar g_{\mathrm{Ca},s}\}$};

\begin{pgfonlayer}{background}
  \node[panel, fill=blue!5, draw=blue!30,
        fit=(vs)(mlp)(aux),
        label={[blue!70,font=\bfseries]above:Neural network}]
        (P1) {};
\end{pgfonlayer}

\draw[arr] (mlp) -- (vs);
\draw[arr] (mlp) -- (aux);

\node[physblock, right=10mm of mlp] (hh)
    {Hodgkin--Huxley\\[1pt]
     \scriptsize $C\dot V = -\!\sum I_{\mathrm{ion}} + g_c(V_s\!-\!V_d)$};
\node[physblock, above=of hh] (autodiff)
    {Automatic differentiation\\[1pt]
     \scriptsize $\partial \widehat{y}/\partial t$};
\node[block, fill=orange!20, draw=orange!70, below=of hh] (res)
    {ODE residual $r(t)$\\[1pt]
     \scriptsize $r = \partial\widehat{y}/\partial t - f(\widehat{y})$};

\begin{pgfonlayer}{background}
  \node[panel, fill=orange!5, draw=orange!40,
        fit=(autodiff)(hh)(res),
        label={[orange!80,font=\bfseries]above:Biophysics}]
        (P2) {};
\end{pgfonlayer}

\draw[arr] (autodiff) -- (hh);
\draw[arr] (hh) -- (res);

\node[lossblock, right=10mm of hh] (ltot)
    {\textbf{Total loss}\\[3pt]
     $\mathcal{L} = \lambda_d\,\mathcal{L}_{V_d}
                  + \lambda_s\,\mathcal{L}^{\mathrm{sparse}}_{V_s}
                  +
                  \lambda_{\mathrm{ph}}\,\mathcal{L}_{\mathrm{phys}}
                  + \lambda_{\mathrm{IC}}\,\mathcal{L}_{\mathrm{IC}}$};

\begin{pgfonlayer}{background}
  \node[panel, fill=red!5, draw=red!40,
        fit=(ltot),
        label={[red!70,font=\bfseries]above:Loss function}]
        (P3) {};
\end{pgfonlayer}

\draw[bigarr] (P0.east) -- (P1.west);
\draw[bigarr] (P1.east) -- (P2.west);
\draw[bigarr] (P2.east) -- (P3.west);

\draw[backprop] (P3.south) -- ++(0,-16mm) -| (P1.south)
    node[pos=0.25, below, font=\scriptsize, black]
    { };

\end{tikzpicture}
}
\caption{Simplified overview of the PINN-based inverse framework.
The time variable and injected dendritic current are provided to the neural
network, while the observed dendritic voltage is enforced through the
data-fidelity term in the loss function.}
\label{fig:pinn_framework}
\end{figure}
\subsection{Inverse problem and network representation}
\label{sec:inverse_problem}

In the inverse setting, the dendritic voltage $V_d(t)$ and the injected dendritic current 
$I_{\mathrm{stim},d}(t)$ are densely observed, the somatic voltage $V_s(t)$ is weakly observed through a small number of sparse somatic samples, and the somatic and dendritic gating variables, together with the calcium concentrations, are treated as fully latent states. The goal is to reconstruct the latent trajectory
\[
y(t)=\left(V_s(t),V_d(t),x_s(t),x_d(t),[\mathrm{Ca}]_1(t),[\mathrm{Ca}]_2(t)\right)
\]
and to estimate the selected somatic maximal conductances
\[
\theta_g=\{\bar g_{Na,s},\bar g_{DR,s},\bar g_{M,s},\bar g_{Ca,s}\}.
\]
Here, $x_s$ and $x_d$ denote the collections of somatic and dendritic gating variables, respectively. 
All remaining biophysical quantities, including the axial coupling $g_c$, the dendritic conductances, the ionic-current expressions, and the voltage-dependent gating kinetics, are fixed according to Tables~\ref{tab:currents}, \ref{tab:parameters} and~\ref{tab:kinetics}. In particular, the PINN does not learn the gating rate functions themselves; the kinetics used in the physics residual are those specified in Table~\ref{tab:kinetics}. The injected dendritic current is treated as a known input.

For each stimulation protocol $r$, the PINN approximates the latent trajectory by a neural network $N_\phi$:
\[
\hat y(t;\phi)=
N_\phi\left(
t_{\mathrm{norm}},
\gamma(t_{\mathrm{norm}}),
I_{\mathrm{norm}}(t),
e_r
\right).
\]
In this expression, $\phi$ denotes all trainable weights and biases of the neural network. The vector $e_r$ is a learned embedding that identifies the stimulation protocol $r$; it allows a single network to represent trajectories generated by different input protocols. The output $\hat y(t;\phi)$ is the predicted full state vector,
\[
\hat y(t;\phi)=
\left(
\hat V_s(t),
\hat V_d(t),
\hat x_s(t),
\hat x_d(t),
[\widehat{\mathrm{Ca}}]_1(t),
[\widehat{\mathrm{Ca}}]_2(t)
\right).
\]
The observed dendritic voltage $V_d^{\mathrm{obs}}(t)$ is not used as a direct input to the network. Instead, it is enforced through the dendritic data-fidelity term in the loss function. Thus, the network receives time, the known injected current, and the protocol identifier as inputs, and it is constrained to produce a dendritic voltage consistent with the observations.

The physical time variable is normalized to the interval $[-1,1]$ by
\[
t_{\mathrm{norm}}
=
2\frac{t-t_{\min}}{t_{\max}-t_{\min}}-1,
\]
where $t_{\min}=0$ ms and $t_{\max}=500$ ms in the simulations reported below. This normalization improves the conditioning of the neural-network approximation. When derivatives with respect to physical time are required, the chain rule gives
\[
\frac{d}{dt}
=
\frac{2}{t_{\max}-t_{\min}}
\frac{d}{dt_{\mathrm{norm}}}.
\]

To improve the representation of sharp action-potential transitions, the normalized time coordinate is mapped to Fourier features:
\[
\gamma(t_{\mathrm{norm}})
=
\left(
\sin(Bt_{\mathrm{norm}}),
\cos(Bt_{\mathrm{norm}})
\right)
\in \mathbb{R}^{2N_f}.
\]
Here, $B=(B_0,\ldots,B_{N_f-1})^\top$ is a vector of prescribed angular frequencies, and the sine and cosine functions are applied componentwise:
\[
\sin(Bt_{\mathrm{norm}})
=
\left(
\sin(B_0t_{\mathrm{norm}}),
\ldots,
\sin(B_{N_f-1}t_{\mathrm{norm}})
\right),
\]
\[
\cos(Bt_{\mathrm{norm}})
=
\left(
\cos(B_0t_{\mathrm{norm}}),
\ldots,
\cos(B_{N_f-1}t_{\mathrm{norm}})
\right).
\]
In the present experiments, we use
\[
B_k=2\pi k,\qquad k=0,\ldots,N_f-1,
\]
with $N_f=16$ Fourier frequency bands. These features help the network represent both slow subthreshold dynamics and the fast voltage changes associated with action potentials.

The injected dendritic current is standardized before being provided to the network:
\[
I_{\mathrm{norm}}(t)
=
\frac{I_{\mathrm{stim},d}(t)-\mu_I}{\sigma_I+\varepsilon},
\]
where $\mu_I$ and $\sigma_I$ are the empirical mean and standard deviation of the injected current over the training data, and $\varepsilon>0$ is a small constant used for numerical stability.

The neural network is an 8-layer multilayer perceptron with 256 hidden units per layer and SiLU activation functions. 
Strictly speaking, the MLP first outputs an unconstrained vector of raw variables,
\[
u(t;\phi)
=
N_\phi\left(
t_{\mathrm{norm}},
\gamma(t_{\mathrm{norm}}),
I_{\mathrm{norm}}(t),
e_r
\right),
\]
where each component of $u(t;\phi)$ takes values in $\mathbb{R}$. These raw outputs are then transformed componentwise into biophysically admissible state variables:
\[
\hat y(t;\phi)=\mathcal T(u(t;\phi)).
\]
The transformation $\mathcal T$ is chosen to enforce the natural ranges of the model variables before they are used in the ODE residual.

First, the somatic and dendritic voltages are mapped through a scaled hyperbolic tangent:
\[
\widehat{V}_k(t)
=
V_c
+
V_a
\tanh\left(
\frac{u_{V_k}(t)}{V_a}
\right),
\qquad k\in\{s,d\},
\]
where $u_{V_k}(t)$ is the raw network output associated with the voltage of compartment $k$. We use $V_c=-20$ mV and $V_a=80$ mV, so that
\[
\widehat{V}_k(t)\in (V_c-V_a,V_c+V_a)=(-100,60)\ \mathrm{mV}.
\]
This interval contains the relevant physiological range of the simulated membrane potentials.

Second, each gating variable is constrained to lie between 0 and 1 using a sigmoid transformation:
\[
\widehat{x}(t)
=
\operatorname{sigmoid}(u_x(t))
=
\frac{1}{1+e^{-u_x(t)}}
\in(0,1),
\]
where $u_x(t)$ is the raw output associated with the gate $x$. This is consistent with the interpretation of gating variables as probabilities or fractions of activated channel subunits.

Third, the calcium concentrations are kept positive and close to the resting concentration by using a softplus transformation:
\[
\widehat{[\mathrm{Ca}]}_i(t)
=
[\mathrm{Ca}]_{\mathrm{rest}}
+
\kappa\,\operatorname{softplus}(u_{\mathrm{Ca},i}(t)),
\qquad i\in\{1,2\},
\]
with
\[
\operatorname{softplus}(z)=\log(1+e^z).
\]
Here, $u_{\mathrm{Ca},i}(t)$ is the raw output associated with calcium pool $i$, and $\kappa=10^{-2}$ scales the calcium fluctuation around the resting value. This guarantees
\[
\widehat{[\mathrm{Ca}]}_i(t)>[\mathrm{Ca}]_{\mathrm{rest}},
\]
while keeping the predicted calcium concentration in a controlled range.

The selected somatic maximal conductances are also constrained to be positive. For each trainable conductance, we use the reparametrization
\[
\bar g_j
=
\operatorname{softplus}(\rho_j)
=
\log(1+e^{\rho_j}),
\qquad
j\in\{Na,DR,M,Ca\},
\]
where the unconstrained variables $\rho_j$ are optimized jointly with the network weights $\phi$. This guarantees positive conductance values throughout training.

The full two-compartment dynamics can be written compactly as
\[
\frac{dy}{dt}
=
f\left(y(t),I_{\mathrm{stim},d}(t);\theta_g\right),
\]
where $f$ denotes the vector field of the conductance-based model. More explicitly, $f$ collects the right-hand sides of the somatic and dendritic voltage equations, the gating equations, and the calcium dynamics:
\[
f(y,I_{\mathrm{stim},d};\theta_g)
=
\left(
f_{V_s},
f_{V_d},
f_{x_s},
f_{x_d},
f_{\mathrm{Ca},1},
f_{\mathrm{Ca},2}
\right).
\]
The voltage components $f_{V_s}$ and $f_{V_d}$ are given by the charge-balance equations in Eq.~(\ref{eq:voltage_dynamics}):
\[
f_{V_s}
=
\frac{1}{C_{m,s}}
\left[
-\bar g_{L,s}(V_s-E_L)
-I_{\mathrm{ion},s}(V_s,x_s;\theta_g)
-\frac{g_c}{p_s}(V_s-V_d)
+I_{\mathrm{stim},s}(t)
\right],
\]
and
\[
f_{V_d}
=
\frac{1}{C_{m,d}}
\left[
-\bar g_{L,d}(V_d-E_L)
-I_{\mathrm{ion},d}(V_d,x_d,[\mathrm{Ca}])
-\frac{g_c}{p_d}(V_d-V_s)
+I_{\mathrm{stim},d}(t)
\right].
\]
In the experiments considered here, stimulation is applied to the dendrite, so $I_{\mathrm{stim},s}(t)=0$. The somatic ionic current $I_{\mathrm{ion},s}$ depends on the trainable conductances $\theta_g$, whereas the dendritic ionic currents and the remaining model parameters are fixed according to Tables~~\ref{tab:currents},\ref{tab:parameters},\ref{tab:kinetics}.

For each gating variable $x$, the corresponding component of the vector field is given by the first-order Hodgkin--Huxley kinetics
\[
f_x(V_k,x)
=
\frac{x_\infty(V_k)-x}{\tau_x(V_k)},
\qquad k\in\{s,d\},
\]
where the functions $x_\infty$ and $\tau_x$ are those reported in Table~\ref{tab:kinetics}. Equivalently, these kinetics can be written in the $\alpha$--$\beta$ form described in Section~\ref{Gating kinetics}. The calcium components $f_{\mathrm{Ca},1}$ and $f_{\mathrm{Ca},2}$ are the right-hand sides of the calcium-pool equations in Eq.~(\ref{eq:ca_dynamics}).

Since $\hat y(t;\phi)$ is differentiable with respect to time, automatic differentiation is used to compute $d\hat y/dt$. The ODE residual is then defined as
\[
r_{\mathrm{ODE}}(t;\phi,\theta_g)
=
\frac{d\hat y(t;\phi)}{dt}
-
f\left(
\hat y(t;\phi),
I_{\mathrm{stim},d}(t);
\theta_g
\right).
\]
If the predicted trajectory exactly satisfies the two-compartment Hodgkin--Huxley model, then 
$r_{\mathrm{ODE}}(t;\phi,\theta_g)=0$ for all $t$. During training, this residual is penalized in the loss function so that the reconstructed trajectory remains consistent with the biophysical equations.
\subsection{Loss function and optimization}
\label{sec:loss_optimization}

The training objective combines dendritic data fidelity, sparse somatic supervision, physics consistency, and initial-condition regularization. For a stimulation protocol indexed by $r$, let $t_i^r$ denote the sampled time points and let $V_d^{\mathrm{obs},r}(t_i^r)$ be the observed dendritic voltage. The observed dendritic dataset is denoted by
\[
\mathcal D_d
=
\left\{
\left(t_i^r,I_{\mathrm{stim},d}^r(t_i^r),V_d^{\mathrm{obs},r}(t_i^r),r\right)
\right\}_{r,i}.
\]
When sparse somatic measurements are available, we also define a binary mask $m_i^r\in\{0,1\}$, where $m_i^r=1$ indicates that a somatic voltage sample is included in the training loss and $m_i^r=0$ otherwise. In the ablation of Section~\ref{sec:res_supervision} the fraction of somatic time points
included in the loss takes the values $0\%$ (i.e.\ $\lambda_s=0$), $1\%$ and $5\%$; the
remaining experiments all use the $5\%$ mask.

The total loss is written as
\[
\mathcal L_{\mathrm{total}}
=
\lambda_d\mathcal L_{V_d}
+
\lambda_s\mathcal L_{V_s}^{\mathrm{sparse}}
+
\lambda_{\mathrm{ph}}\mathcal L_{\mathrm{phys}}
+
\lambda_{\mathrm{IC}}\mathcal L_{\mathrm{IC}}.
\]
The dendritic data term enforces agreement between the predicted dendritic voltage and the observed dendritic trace:
\[
\mathcal L_{V_d}
=
\frac{1}{N_d}
\sum_{r,i}
\left(
\widehat V_d^r(t_i^r)
-
V_d^{\mathrm{obs},r}(t_i^r)
\right)^2,
\]
where $N_d$ is the total number of dendritic observation points across all protocols.

The sparse somatic term provides weak direct supervision of the hidden somatic voltage:
\[
\mathcal L_{V_s}^{\mathrm{sparse}}
=
\frac{1}{N_s}
\sum_{r,i}
m_i^r
\left(
\widehat V_s^r(t_i^r)
-
V_s^{\mathrm{sparse},r}(t_i^r)
\right)^2,
\qquad
N_s=\sum_{r,i}m_i^r.
\]
This term is not a dense somatic supervision term. It only anchors the hidden somatic trajectory
at a small number of scattered time points. Setting $\lambda_s=0$ recovers the strictly
dendrite-only inverse setting, which is the $0\%$ configuration of the ablation; in all other
experiments $\lambda_s>0$ is used to weakly regularize the reconstruction.

The physics loss penalizes the ODE residuals defined in Section~\ref{sec:inverse_problem}. Let
\[
r_{\mathrm{ODE}}(t;\phi,\theta_g)
=
\frac{d\widehat y(t;\phi)}{dt}
-
f\left(
\widehat y(t;\phi),
I_{\mathrm{stim},d}(t);
\theta_g
\right),
\]
and let $r_q(t;\phi,\theta_g)$ denote the component of this residual associated with state variable $q$. The set $Q$ contains all state components used in the physics residual, including voltage, gating, and calcium variables. We use a Huber penalty to reduce the influence of large residuals near spike onsets:
\[
\mathcal L_{\mathrm{phys}}
=
\sum_{q\in Q}
w_q^{(\mathrm{step})}
\frac{1}{N_{\mathrm{ph}}}
\sum_{r,i}
\rho_\delta
\left(
r_q^r(t_i^r;\phi,\theta_g)
\right),
\]
where $N_{\mathrm{ph}}$ is the number of residual evaluation points and
\[
\rho_\delta(z)
=
\begin{cases}
\dfrac{z^2}{2\delta}, & |z|\leq \delta,\\[6pt]
|z|-\dfrac{\delta}{2}, & |z|>\delta.
\end{cases}
\]
In all experiments, we use $\delta=25$. The physics weights are gradually increased during training in order to avoid stiff ODE residuals dominating the early optimization stage:
\[
w_q^{(\mathrm{step})}
=
\alpha(\mathrm{step})w_q^{(0)},
\qquad
\alpha(\mathrm{step})
=
\min\left(1,\frac{\mathrm{step}}{N_{\mathrm{warmup}}}\right),
\]
with $N_{\mathrm{warmup}}=10{,}000$. The global physics weight $\lambda_{\mathrm{ph}}$ can be absorbed into the component weights $w_q^{(0)}$; equivalently, the time-dependent weights $w_q^{(\mathrm{step})}$ control the relative contribution of the physics residuals during training.

The initial-condition loss stabilizes training and guides the predicted trajectory toward physiologically plausible initial states. It is computed from the first $n_{\mathrm{IC}}=20$ samples of each run and is written schematically as
\[
\mathcal L_{\mathrm{IC}}
=
w_{\mathrm{IC},V_d}\mathcal L_{\mathrm{IC},V_d}
+
w_{\mathrm{IC,gates}}\mathcal L_{\mathrm{IC,gates}}
+
w_{\mathrm{IC},V_s}\mathcal L_{\mathrm{IC},V_s}
+
w_{\mathrm{IC,soma\_g}}\mathcal L_{\mathrm{IC,soma\_g}}.
\]
The dendritic-voltage component anchors the initial dendritic voltage near the observed value, the gating-variable component biases the gates toward their voltage-dependent steady states, and the calcium component biases calcium concentrations toward their resting values. A weak prior on the initial somatic voltage is also used to improve stability when the soma is only sparsely observed.

The network weights $\phi$ and the unconstrained conductance variables $\rho_j$ are optimized jointly using the Adam optimizer. Unless otherwise stated, we use learning rate $10^{-4}$, mini-batches of $2048$ randomly sampled time points across all protocols, and $30{,}000$ optimization steps. The main loss weights are fixed to $\lambda_d=1$ for the dendritic voltage term, $\lambda_s=1.5$ for the sparse somatic term, and $\lambda_{\mathrm{IC}}=0.2$ for the initial-condition regularization. The component weights used in the initial-condition loss are
\[
w_{\mathrm{IC},V_d}=1,\qquad
w_{\mathrm{IC,gates}}=0.1,\qquad
w_{\mathrm{IC},V_s}=1,\qquad
w_{\mathrm{IC,soma\_g}}=1
\]
\subsection{Synthetic data and robustness protocols}
\label{sec:synthetic_data}

Synthetic datasets were generated by forward simulation of the full two-compartment model described in Section~\ref{sec:biophysical_model}. The simulations were performed using MATLAB's stiff ODE solver \texttt{ode15s} over the interval $t\in[0,500]$ ms. In all inverse experiments, stimulation was applied to the dendritic compartment only, so that
\[
I_{\mathrm{stim},s}(t)=0.
\]
The dendritic stimulation current $I_{\mathrm{stim},d}(t)$ was treated as a known input and was not estimated by the PINN.

We used four dendritic stimulation protocols: constant step, ramp, rectified sinusoidal input, and pulse train. These protocols were chosen to excite complementary dynamical regimes of the two-compartment neuron model. The constant input probes the response to a sustained depolarizing drive, the ramp input probes gradual recruitment of firing, the sinusoidal input probes oscillatory forcing, and the pulse train probes repeated transient responses.

Let $t_{\mathrm{on}}$ and $t_{\mathrm{off}}$ denote the onset and offset times of the stimulation window. In the simulations reported below, the stimulation window is $[t_{\mathrm{on}},t_{\mathrm{off}}]=[100,400]$ ms. The four injected-current protocols are defined as follows.

The constant step protocol is
\[
I_{\mathrm{stim},d}^{\mathrm{const}}(t)
=
\begin{cases}
I_{\mathrm{amp}}, & t\in[t_{\mathrm{on}},t_{\mathrm{off}}],\\
0, & \text{otherwise}.
\end{cases}
\]

The ramp protocol is
\[
I_{\mathrm{stim},d}^{\mathrm{ramp}}(t)
=
\begin{cases}
I_{\mathrm{amp}}
\dfrac{t-t_{\mathrm{on}}}{t_{\mathrm{off}}-t_{\mathrm{on}}},
& t\in[t_{\mathrm{on}},t_{\mathrm{off}}],\\[8pt]
0, & \text{otherwise}.
\end{cases}
\]
This input increases linearly from $0$ to $I_{\mathrm{amp}}$ during the stimulation window.

The rectified sinusoidal protocol is
\[
I_{\mathrm{stim},d}^{\mathrm{sin}}(t)
=
\begin{cases}
\dfrac{I_{\mathrm{amp}}}{2}
\left[
1+
\sin\left(
2\pi f
\dfrac{t-t_{\mathrm{on}}}{T_{\mathrm{ref}}}
\right)
\right],
& t\in[t_{\mathrm{on}},t_{\mathrm{off}}],\\[8pt]
0, & \text{otherwise},
\end{cases}
\]
where $f$ controls the number of oscillatory cycles and $T_{\mathrm{ref}}$ is a reference time scale. In the generated dataset, $T_{\mathrm{ref}}=500$ ms was used. The factor $1/2$ and the shift by $1$ ensure that the current remains non-negative.

The pulse-train protocol is
\[
I_{\mathrm{stim},d}^{\mathrm{pulse}}(t)
=
I_{\mathrm{amp}}
\sum_{\ell=1}^{N_p}
\mathbf 1_{[t_\ell,t_\ell+\tau_{\mathrm{pulse}}]}(t),
\]
where $\mathbf 1_A(t)$ is the indicator function of the interval $A$, $N_p$ is the number of pulses, $\tau_{\mathrm{pulse}}$ is the pulse duration, and
\[
t_\ell=t_{\mathrm{start}}+(\ell-1)\Delta t
\]
is the onset time of pulse $\ell$. Here, $\Delta t$ denotes the inter-pulse interval.

The amplitudes used in the four protocols were chosen from preliminary forward simulations so that the weakly coupled model generated informative somatic and dendritic dynamics without saturating the voltage response. In the clean benchmark shown in Figure~\ref{fig:clean_protocols}, the amplitudes are $I_{\mathrm{amp}}=15$ for the constant input, $I_{\mathrm{amp}}=20$ for the ramp input, $I_{\mathrm{amp}}=7$ for the sinusoidal input, and $I_{\mathrm{amp}}=10$ for the pulse train. The resulting trajectories provide the reference somatic and dendritic voltages used for training and evaluation.

Figure~\ref{fig:clean_protocols} displays the four stimulation protocols together with the corresponding somatic voltage $V_s(t)$, dendritic voltage $V_d(t)$, and injected dendritic current $I_{\mathrm{stim},d}(t)$. The full somatic voltage trajectory is used only for evaluation, except for the small subset of sparse somatic samples included in the training loss.

To evaluate robustness to measurement noise, additional synthetic datasets were generated by adding Gaussian noise to the voltage observations:
\[
V_k^{\mathrm{obs},\sigma}(t_i)
=
V_k^{\mathrm{clean}}(t_i)
+
\sigma \xi_{k,i},
\qquad
\xi_{k,i}\sim\mathcal N(0,1),
\]
where $k\in\{s,d\}$ denotes the somatic or dendritic compartment and $\sigma$ is the noise standard deviation in mV. We considered $\sigma=0.5$, $1$, and $2$ mV, in addition to the clean case $\sigma=0$. For each noise level, the same four dendritic stimulation protocols were used. The architecture, loss weights, optimizer, batch size, number of training iterations, and sparse somatic sampling mask were kept identical across noise conditions. Sparse somatic samples were extracted from the noisy somatic trace using the same mask as in the clean-data experiment, while the full clean somatic trace was reserved as the reference trajectory for evaluation.

To assess sensitivity to stochastic initialization, the clean-data experiment was repeated over eight independent random seeds. For each run, the network was trained from scratch on the same clean dataset, using the same stimulation protocols and the same sparse somatic sampling mask. All methodological settings were kept identical; only the random seed controlling neural-network initialization and stochastic mini-batch sampling was changed. Each trained model was then evaluated on the full clean dataset, including the complete somatic voltage trajectory reserved for evaluation. Across the eight runs, we report the mean, standard deviation, min--max range, coefficient of variation, and relative error for the reconstructed voltages, training loss, and estimated conductances.
\begin{figure}[h!]
  \centering
  \includegraphics[width=0.8\linewidth]{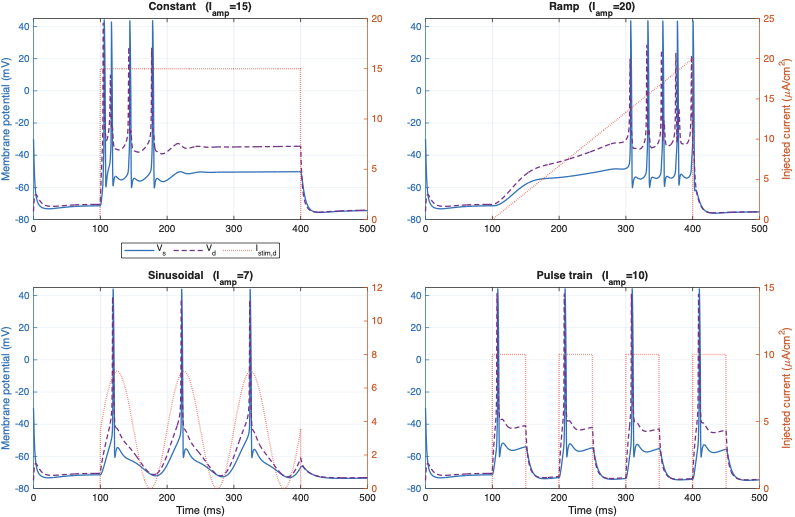}
 \caption{Forward simulations of the coupled soma--dendrite voltage system~\eqref{eq:voltage_dynamics} 
under four dendritic stimulation protocols: constant step, ramp, rectified sinusoidal input, and pulse train. 
The solid blue curve represents the somatic voltage $V_s(t)$, the dashed purple curve represents the 
dendritic voltage $V_d(t)$, and the red dotted curve represents the injected dendritic current 
$I_{\mathrm{stim},d}(t)$. 
Voltages are plotted on the left axis in mV, and the injected current on the right axis in 
$\mu\mathrm{A}/\mathrm{cm}^2$.}
\label{fig:clean_protocols}
\end{figure}

\section{Results}
\label{sec:results}

The numerical experiments were designed to test whether a PINN can recover the somatic output of a weakly coupled two-compartment neuron from dense dendritic measurements complemented by a small fraction of somatic samples. We focused on the weak-coupling regime $g_c=0.1~\mathrm{mS/cm^2}$, because in this regime the soma and dendrite do not behave as a single equipotential unit. As shown by the coupling analysis in Figure~\ref{fig:coupling_synchronization} and by the voltage-balance system~\eqref{eq:voltage_dynamics}, the axial current only weakly transfers fast somatic events to the dendrite. Consequently, the dendritic voltage $V_d(t)$ is a filtered and attenuated representation of the somatic voltage $V_s(t)$, rather than a direct copy of the neuronal output. This makes the reconstruction of $V_s(t)$ from dendritic observations a genuinely partially observed inverse problem.

All synthetic data used in the inverse experiments were generated by forward simulation of the full two-compartment Hodgkin--Huxley model described in Section~\ref{sec:biophysical_model}, using the ionic currents in Table~\ref{tab:currents}, the passive parameters and maximal conductances in Table~\ref{tab:parameters}, and the voltage-dependent gating kinetics in Table~\ref{tab:kinetics}. The forward simulations were performed in MATLAB, and the dendritic stimulation current $I_{\mathrm{stim},d}(t)$ was treated as a known input. Four dendritic stimulation protocols were used to enrich the training data and excite different dynamical regimes of the model: a constant step input, a ramp input, a rectified sinusoidal input, and a pulse-train input (Figure~\ref{fig:clean_protocols}). The objective was not to infer the stimulation current, but to reconstruct the hidden somatic trajectory and estimate a small set of somatic maximal conductances from the observed dendritic response.

The training setup was weakly supervised. The PINN was trained mainly from dense dendritic voltage observations and the known injected dendritic current, but a small number of scattered somatic voltage samples were also included in the loss function. The fraction of somatic time points included in the loss is the main experimental variable of this section. We first ask whether the somatic trajectory can be recovered at all without somatic samples ($0\%$, i.e.\ $\lambda_s=0$), and then quantify how the reconstruction and the conductance estimates improve when $1\%$ and $5\%$ of the somatic time points are added to the loss (Sections~\ref{sec:res_supervision} and~\ref{sec:res_conductances}). The $5\%$ configuration, which is the largest amount of somatic information considered in this work, is then kept unchanged for the comparison with the unscented Kalman filter (Section~\ref{sec:ukf}) and for the robustness experiments (Section~\ref{sec:res_robustness}). In all cases, the complete somatic voltage trace was reserved for evaluation. Therefore, the results should be interpreted as reconstruction from dendritic observations with sparse somatic guidance, rather than as a fully supervised somatic reconstruction. The use of Fourier features in the PINN was motivated by the time-series structure of the problem: the voltage traces contain both slow subthreshold components and fast spike transitions. Standard neural networks are known to learn low-frequency components more easily than high-frequency ones, a phenomenon often described as spectral bias~\cite{rahaman2019spectral,xu2019frequency}. Fourier feature embeddings help mitigate this limitation by enriching the time representation with oscillatory components~\cite{tancik2020}.

\subsection{Voltage reconstruction versus the amount of somatic supervision}

\label{sec:res_supervision}

Figure~\ref{fig:clean_protocols} shows the clean synthetic trajectories generated under the four dendritic stimulation protocols. In each panel, the solid blue curve represents the somatic voltage $V_s(t)$, the dashed purple curve represents the dendritic voltage $V_d(t)$, and the red dotted curve represents the injected dendritic current $I_{\mathrm{stim},d}(t)$. The figure illustrates the main difficulty of the inverse problem. The somatic voltage contains sharp action potentials, whereas the dendritic voltage is smoother and more attenuated, especially during spike events. This difference is a direct consequence of weak soma--dendrite coupling and of the fact that the two compartments have distinct ionic compositions and state variables.

The four protocols in Figure~\ref{fig:clean_protocols} probe complementary aspects of the dynamics. The constant input produces a short spike train followed by a depolarized plateau. The ramp input gradually recruits firing as the dendritic drive increases. The rectified sinusoidal input induces phase-dependent spike events and subthreshold oscillations. The pulse train produces repeated transient responses to abrupt current pulses. Using these different protocols in a single training framework exposes the PINN to both smooth and abrupt input-output regimes, which is important for learning a robust latent representation of the soma--dendrite dynamics.
\begin{figure}[htbp]
    \centering
    \includegraphics[width=0.7\textwidth]{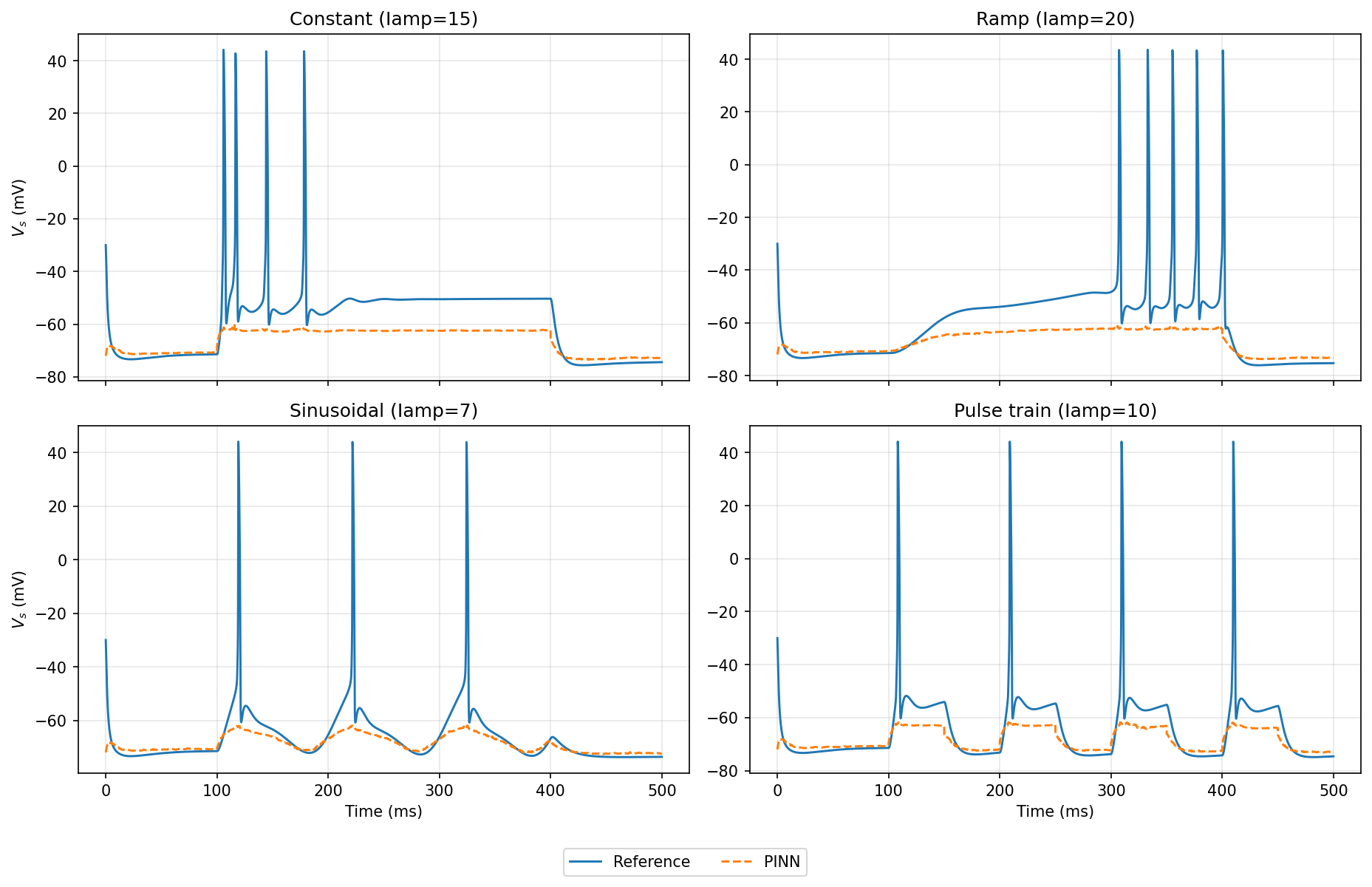}
    \caption{\textbf{Somatic voltage reconstruction with 0\% sparse somatic supervision (clean data).}
    Comparison of the ground-truth somatic membrane potential $V_s(t)$
    (blue, reference forward simulation) and the PINN reconstruction
    (orange dashed) for the four dendritic stimulation protocols over
    $t \in [0,500]$\,ms. The PINN is trained using dense dendritic
    observations only, with no somatic voltage samples included in the loss ($\lambda_s=0$).}
    \label{fig:somatic_reconstruction_0percent}
\end{figure}

\begin{figure}[htbp]
    \centering
    \includegraphics[width=0.7\textwidth]{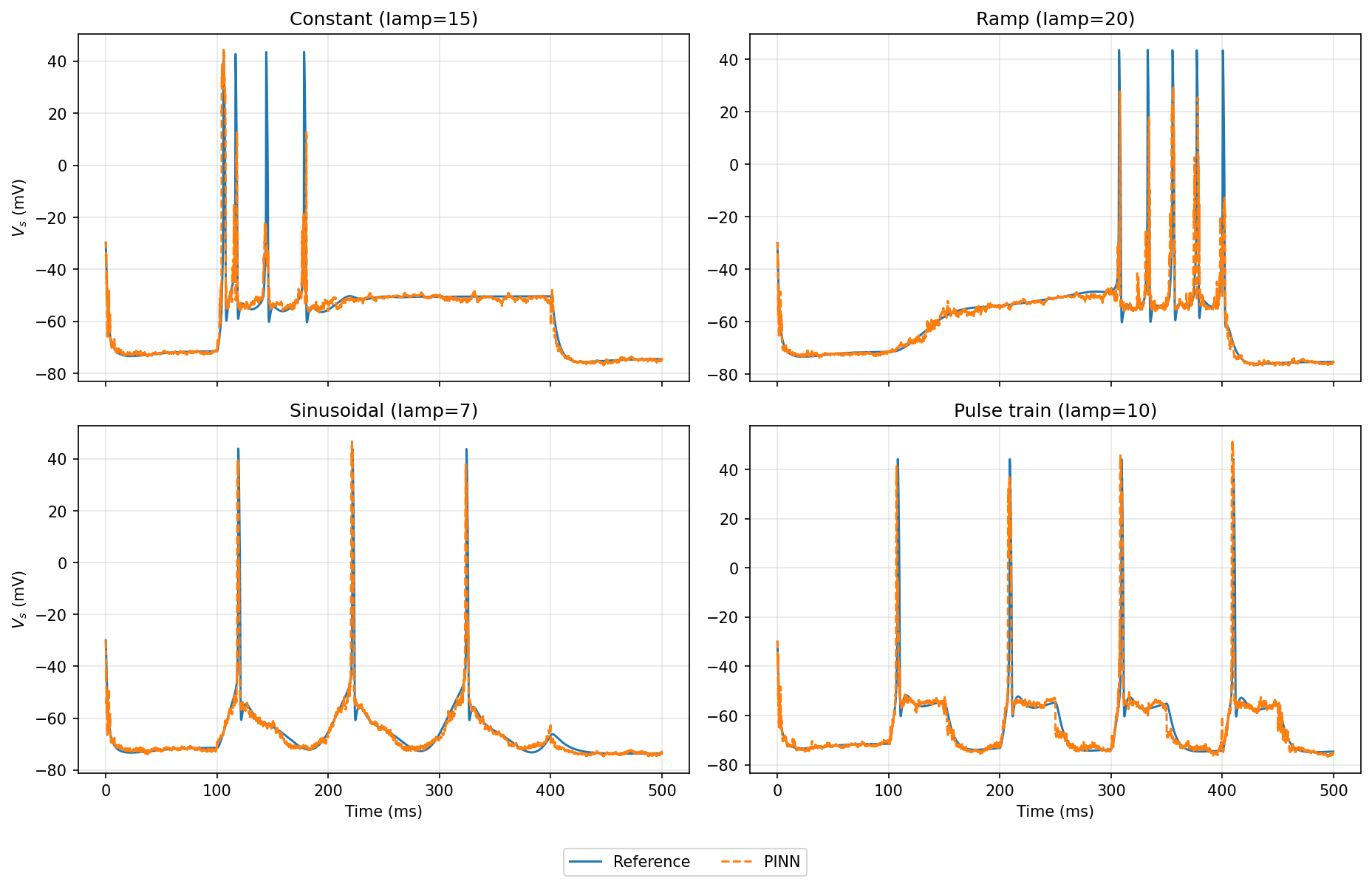}
    \caption{\textbf{Somatic voltage reconstruction with 1\% sparse somatic supervision (clean data).}
    Comparison of the ground-truth somatic membrane potential $V_s(t)$
    (blue, reference forward simulation) and the PINN reconstruction
    (orange dashed) for the four dendritic stimulation protocols over
    $t \in [0,500]$\,ms. The PINN is trained using dense dendritic
    observations together with only 1\% sparsely sampled somatic voltage measurements.}
    \label{fig:somatic_reconstruction_1percent}
\end{figure}

\paragraph{No somatic supervision ($0\%$).}
Figure~\ref{fig:somatic_reconstruction_0percent} shows the reconstruction obtained in the strictly dendrite-only setting, i.e.\ with $\lambda_s=0$ in the loss of Section~\ref{sec:loss_optimization}. In all four protocols the network returns a smooth, nearly flat somatic trajectory that hardly responds to the stimulation: the depolarized plateau of the constant protocol, the slow depolarization of the ramp and the subthreshold oscillations of the sinusoidal input are only coarsely reproduced, and none of the action potentials is recovered. The $0\%$ column of Table~\ref{tab:ablation_sparse_compact} makes this quantitative. The MAE lies between $3.78$~mV (sinusoidal) and $8.43$~mV (constant), the RMSE between $10.29$ and $14.28$~mV, and the maximum error is close to $106$~mV in every protocol, which is the amplitude of a missed spike. The ordering of the protocols is instructive: the largest errors are obtained for the constant and ramp inputs, whose somatic response contains a sustained depolarization, whereas the sinusoidal input, for which $V_s$ returns close to rest between the three spikes, gives the smallest MAE. At $0\%$ the error is therefore not only due to the missing spikes but also to an offset of the subthreshold level.

This behaviour can be understood from the structure of the loss in the weak-coupling regime. The only place where $V_s$ enters the dendritic current balance of Eq.~\eqref{eq:voltage_dynamics} is the axial term $(g_c/p_d)(V_d-V_s)$, which is small for $g_c=0.1~\mathrm{mS/cm^2}$. A non-spiking somatic trajectory therefore leaves only a modest residual in the dendritic equation, and this residual, weighted as in Section~\ref{sec:loss_optimization} (Huber penalty, progressive activation of the physics weights), is not sufficient to drive the network away from the smooth solution towards which it is biased by the spectral properties of neural networks~\cite{rahaman2019spectral,xu2019frequency}. This is consistent with the failure modes reported for physics-informed networks on problems in
which the differential-equation residual alone has to select a non-trivial solution, where the
optimizer converges instead to a smooth, nearly trivial
trajectory~\cite{krishnapriyan2021failure}. The optimizer thus settles on a solution in which the somatic equation is satisfied by a soma with reduced excitability (see the corresponding conductance estimates in Section~\ref{sec:res_conductances}), rather than on the spiking solution. The result characterizes the present formulation of the loss, not an absolute limit of the dendritic information, a point we return to in the Conclusion; within this formulation, the somatic samples are needed to remove the ambiguity.

\paragraph{One percent of somatic samples.}
With $1\%$ of the somatic time points included in the loss (Figure~\ref{fig:somatic_reconstruction_1percent}), the reconstruction changes qualitatively. All action potentials now appear at approximately the correct times, the plateau of the constant protocol and the subthreshold trajectories of the other protocols are followed, and the MAE drops to $1.92$--$2.70$~mV (Table~\ref{tab:ablation_sparse_compact}, $1\%$ column). The RMSE ($6.64$--$9.74$~mV) and the maximum error ($81.7$--$93.9$~mV) nevertheless remain large. The two metrics disagree because they weight the errors differently: the MAE is dominated by the many subthreshold points, which are now well reconstructed, whereas the RMSE and $E_{\max}$ are dominated by the few points located on the spikes, several of which are still reproduced with a reduced amplitude or a small temporal shift; the reconstruction also shows small oscillations in the subthreshold intervals. A convenient indicator is the ratio RMSE/MAE, which increases when the error is concentrated on a few points: it lies between $3.5$ and $4.1$ at $1\%$, against $2.0$--$2.3$ at $5\%$. A handful of somatic anchors is thus sufficient to select the spiking branch of the solution, but not to pin down the amplitude and the timing of every event.

\paragraph{Five percent of somatic samples.}

Figure~\ref{fig:pinn-clean} presents the main clean-data reconstruction result\rev{, obtained
with $5\%$ of the somatic time points}.

\begin{figure}[h!]
    \centering
    \includegraphics[width=0.7\textwidth]{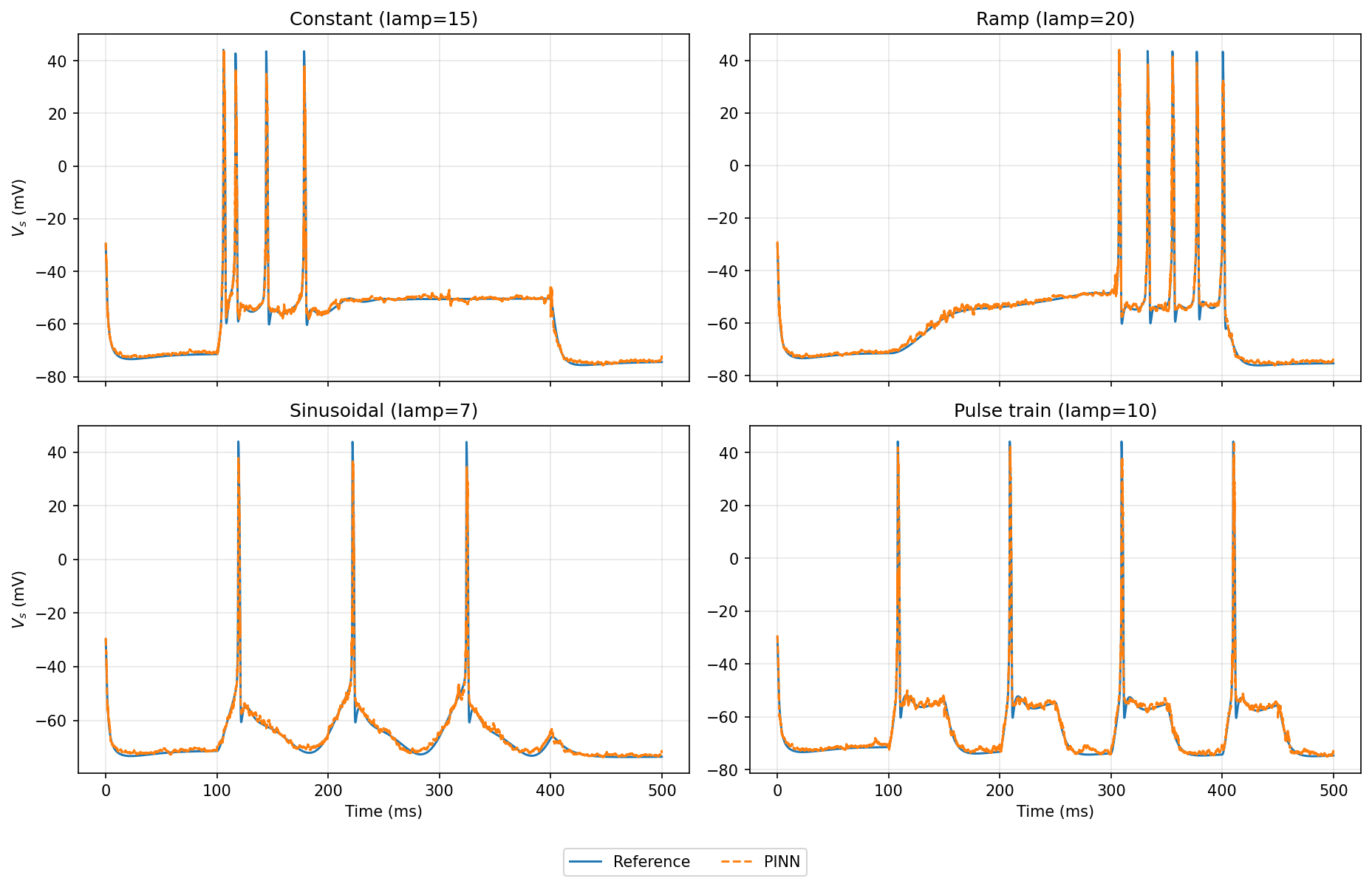}
     \caption{\textbf{Somatic voltage reconstruction with 5\% sparse somatic supervision
(clean data).} Comparison of the ground-truth somatic membrane potential $V_s(t)$
(blue, reference forward simulation) and the PINN reconstruction (orange dashed) for the four
dendritic stimulation protocols over $t \in [0, 500]$\,ms.}

\label{fig:pinn-clean}
\end{figure}
The blue curves show the reference somatic voltage obtained from the forward simulation, and the orange dashed curves show the PINN reconstruction. Across all four protocols, the reconstructed voltage closely follows the reference trajectory. For the constant input, the network recovers the spike timing, spike amplitude, and afterhyperpolarization phases following stimulation onset. For the ramp input, it captures both the slow depolarization before firing and the subsequent repetitive spike train. For the sinusoidal input, it reproduces the subthreshold oscillations and the spike events induced by the periodic drive. For the pulse-train input, it reconstructs the repeated spike responses triggered by the current pulses and the recovery phases between them.

This is the point of the experiment: the dense observed signal is dendritic, whereas the
reconstructed signal is the hidden somatic voltage. The close agreement in Figure~\ref{fig:pinn-clean} shows that the PINN can use the two-compartment biophysical constraints to propagate information from the observed dendrite to the weakly observed soma. The sparse somatic samples provide local anchor points for $V_s(t)$, but the full continuous trajectory, including most spike peaks and inter-spike intervals, still has to be inferred from the dendritic trace and the Hodgkin--Huxley equations.

The pointwise reconstruction error is quantified in Figure~\ref{fig:vs_abs_error}, which superimposes, for each protocol, the absolute error obtained with $0\%$, $1\%$ and $5\%$ somatic supervision.
For each stimulation protocol, we compute the absolute error between the reference somatic voltage and the PINN reconstruction:
\[
e_{\mathrm{abs}}(t)
=
\left|V_s^{\mathrm{ref}}(t)-V_s^{\mathrm{PINN}}(t)\right|.
\]
To summarize the reconstruction accuracy, we report three metrics. The mean absolute error is defined as
\[
\mathrm{MAE}
=
\frac{1}{N}
\sum_{i=1}^{N}
\left|V_s^{\mathrm{ref}}(t_i)-V_s^{\mathrm{PINN}}(t_i)\right|,
\]
the root-mean-square error is defined as
\[
\mathrm{RMSE}
=
\sqrt{
\frac{1}{N}
\sum_{i=1}^{N}
\left(
V_s^{\mathrm{ref}}(t_i)-V_s^{\mathrm{PINN}}(t_i)
\right)^2
},
\]
and the maximum absolute error is
\[
E_{\max}
=
\max_i
\left|V_s^{\mathrm{ref}}(t_i)-V_s^{\mathrm{PINN}}(t_i)\right|.
\]
Here, \(N\) denotes the number of evaluation time points for a given stimulation protocol.
\begin{figure}[h!]
    \centering
    \includegraphics[width=0.7\textwidth]{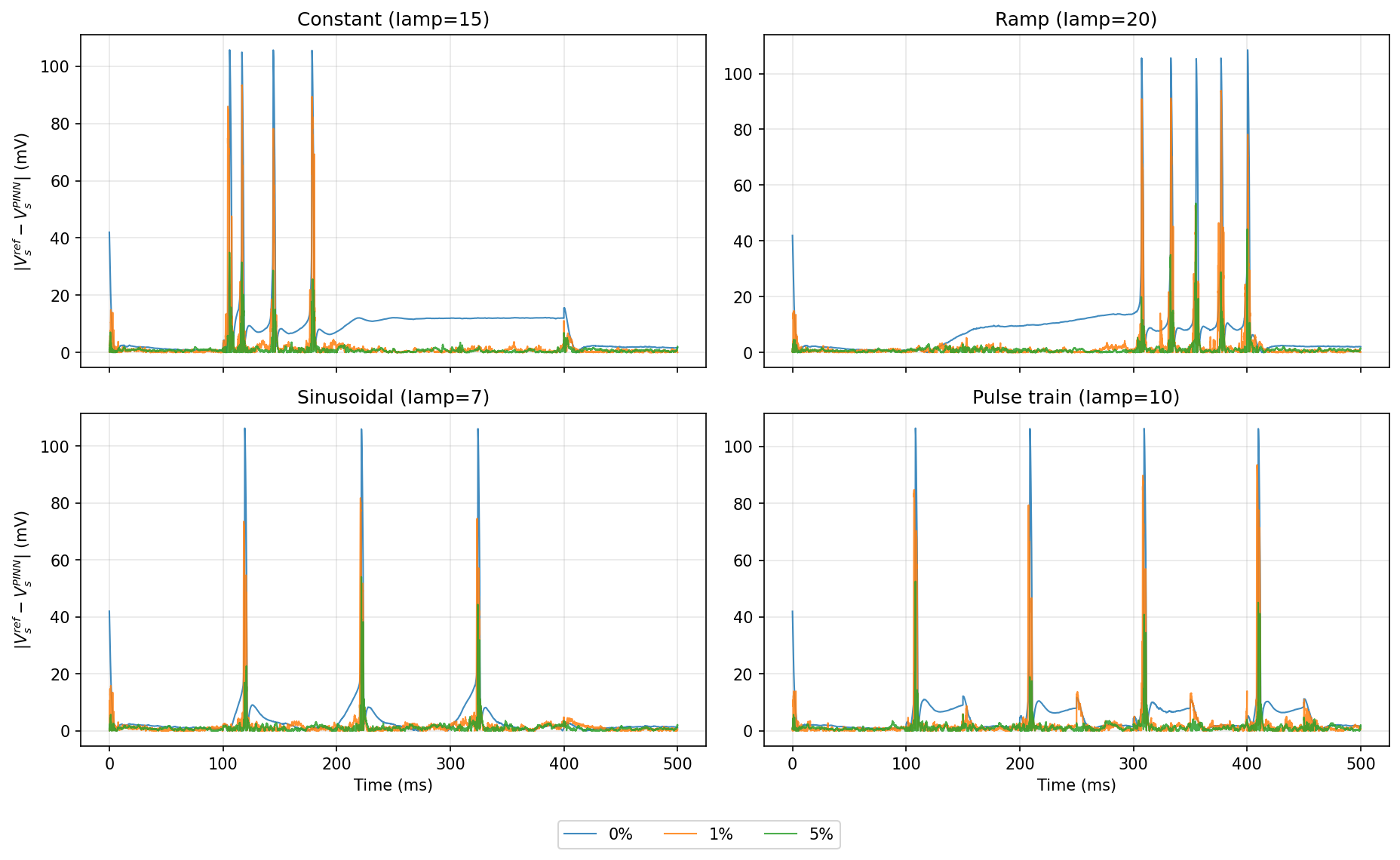}

   \caption{Pointwise somatic voltage reconstruction error $e_{\mathrm{abs}}(t)$ for the four stimulation protocols and for the three levels of somatic supervision ($0\%$, $1\%$ and $5\%$ of the somatic time points; same colour code in all panels).
Each panel shows the time evolution of the absolute difference between the reference somatic voltage and the PINN reconstruction.
The corresponding MAE, RMSE and maximum absolute errors are listed in Table~\ref{tab:ablation_sparse_compact}.}
\label{fig:vs_abs_error}
\end{figure}
The error remains low during most subthreshold intervals but increases around spike peaks and rapid voltage transitions. This is expected because a small temporal shift during the fast upstroke or repolarization of an action potential can produce a large instantaneous voltage difference. The protocol-wise errors confirm this interpretation.

With $5\%$ supervision (Table~\ref{tab:ablation_sparse_compact}, last three columns), the constant protocol gives $\mathrm{MAE}=0.90$~mV and $\mathrm{RMSE}=1.79$~mV; the ramp protocol gives $\mathrm{MAE}=1.03$~mV and $\mathrm{RMSE}=2.37$~mV; the sinusoidal protocol gives $\mathrm{MAE}=1.10$~mV and $\mathrm{RMSE}=2.23$~mV; and the pulse train, together with the ramp, is the most difficult case, with $\mathrm{MAE}=1.20$~mV and $\mathrm{RMSE}=2.37$~mV. The larger error for the pulse train reflects the difficulty of resolving abrupt repeated inputs and closely timed fast transitions. Overall, Figure~\ref{fig:vs_abs_error} confirms quantitatively the visual reconstruction quality observed in Figure~\ref{fig:pinn-clean}.
\begin{table}[htbp]
\centering
\caption{Somatic reconstruction error versus the fraction of somatic samples used in the loss
(identical hyperparameters and seed across all runs). Values in mV.}
\label{tab:ablation_sparse_compact}

\begin{tabular}{
l
S[table-format=1.2] S[table-format=2.2] S[table-format=3.1]
S[table-format=1.2] S[table-format=2.2] S[table-format=3.1]
S[table-format=1.2] S[table-format=2.2] S[table-format=3.1]
}
\toprule
& \multicolumn{3}{c}{$0\%$}
& \multicolumn{3}{c}{$1\%$}
& \multicolumn{3}{c}{$5\%$} \\
\cmidrule(lr){2-4}
\cmidrule(lr){5-7}
\cmidrule(lr){8-10}

Protocol
& {MAE} & {RMSE} & {$E_{\max}$}
& {MAE} & {RMSE} & {$E_{\max}$}
& {MAE} & {RMSE} & {$E_{\max}$} \\
\midrule

Constant
& 8.43 & 13.68 & 105.7
& 2.02 & 8.24 & 93.4
& 0.90 & 1.79 & 34.9 \\

Ramp
& 7.83 & 14.28 & 108.5
& 2.15 & 8.02 & 93.9
& 1.03 & 2.37 & 53.5 \\

Sinusoidal
& 3.78 & 10.29 & 106.2
& 1.92 & 6.64 & 81.7
& 1.10 & 2.23 & 54 \\

Pulse train
& 5.54 & 12.22 & 106.4
& 2.70 & 9.74 & 93.5
& 1.20 & 2.37 & 52.5 \\

\bottomrule
\end{tabular}
\end{table}

Table~\ref{tab:ablation_sparse_compact} gathers the three levels of supervision for the four protocols. Read column by column, the RMSE decreases by a factor between $4.6$ (sinusoidal) and $7.6$ (constant) from $0\%$ to $5\%$, and by a factor between $3.0$ and $4.6$ from $1\%$ to $5\%$; the maximum error falls from about $106$~mV to $35$--$54$~mV. The residual $E_{\max}$ at $5\%$ deserves a comment, because every spike is recovered at this level (see also Section~\ref{sec:ukf}). The upstroke of an action potential in this model has a slope of the order of a few hundred mV/ms, so that a shift of one or two tenths of a millisecond between the reconstructed and the reference spike produces an instantaneous error of several tens of mV without any event being missed. $E_{\max}$ is thus a measure of spike-timing precision rather than of spike detection, and the MAE and RMSE, which average over the whole window, are the relevant metrics for the quality of the trajectory. In this sense the transition from $0\%$ to $5\%$ is a change in the nature of the error: from a globally wrong trajectory (missing spikes and biased subthreshold level) to a correct trajectory with sub-millisecond timing errors concentrated on the spikes.

\subsection{Conductance identification }
\label{sec:res_conductances}
\begin{table}[ht]
\centering
\caption{Parameter estimation for different levels of somatic supervision.}
\label{tab:somatic_supervision_ablation}
\resizebox{\textwidth}{!}{%
\begin{tabular}{lc cc cc cc}
    \toprule
    & \textbf{Truth}
    & \multicolumn{2}{c}{\textbf{0\% supervision}}
    & \multicolumn{2}{c}{\textbf{1\% supervision}}
    & \multicolumn{2}{c}{\textbf{5\% supervision}} \\
    \cmidrule(lr){3-4}
    \cmidrule(lr){5-6}
    \cmidrule(lr){7-8}

    \textbf{Parameter}
    & \textbf{Value}
    & \textbf{Estimated} & \textbf{Rel. error}
    & \textbf{Estimated} & \textbf{Rel. error}
    & \textbf{Estimated} & \textbf{Rel. error} \\
    \midrule

    $\bar g_{\mathrm{Na},s}$
    & 120
    & 51.907 & 56.74\%
    & 119.908 & 0.08\%
    & 119.930 & 0.06\% \\

    $\bar g_{\mathrm{DR},s}$
    & 12
    & 7.149 & 40.42\%
    & 12.113 & 0.94\%
    & 11.990 & 0.08\% \\

    $\bar g_{\mathrm{M},s}$
    & 0.25
    & 0.005 & 97.81\%
    & 0.052 & 79.14\%
    & 0.206 & 17.77\% \\

    $\bar g_{\mathrm{Ca},s}$
    & 0.10
    & 0.672 & 572.48\%
    & 0.168 & 67.74\%
    & 0.114 & 14.11\% \\

    \bottomrule
\end{tabular}%
}
\end{table}
In addition to reconstructing the hidden somatic voltage, the PINN was used to estimate the selected somatic maximal conductances
\[
\theta_g=\{\bar g_{Na,s},\bar g_{DR,s},\bar g_{M,s},\bar g_{Ca,s}\}.
\]

The conductance identification results for the three supervision levels, including the clean-data case at $5\%$, are reported in Table~\ref{tab:somatic_supervision_ablation}.
For each estimated conductance, the relative error is computed as
\[
\mathrm{Rel.\ error}
=
\frac{
\left|\bar g_{\mathrm{true}}-\bar g_{\mathrm{est}}\right|
}{
\left|\bar g_{\mathrm{true}}\right|
}
\times 100\%.
\]
Here, $\bar g_{\mathrm{true}}$ denotes the reference conductance used to generate the synthetic data, and
$\bar g_{\mathrm{est}}$ denotes the value estimated by the PINN.

The three columns of Table~\ref{tab:somatic_supervision_ablation} follow the same progression as the voltage errors of Table~\ref{tab:ablation_sparse_compact}. Without somatic samples, the estimates are far from the reference values: $\bar g_{Na,s}$ and $\bar g_{DR,s}$ are underestimated by $57\%$ and $40\%$, $\bar g_{M,s}$ collapses to almost zero ($0.005~\mathrm{mS/cm^2}$) and $\bar g_{Ca,s}$ is overestimated by a factor of almost seven. These values are not arbitrary: they are the conductances of a low-excitability soma consistent with the smooth trajectory of Figure~\ref{fig:somatic_reconstruction_0percent}, in which the fast currents are reduced and the slow depolarizing balance is adjusted by the M-type and calcium terms. The parameters are co-adapted to a wrong trajectory, and the loss, as weighted in the present formulation, does not penalize this (trajectory, parameter) pair strongly enough to reject it. This is the parameter-side counterpart of the ambiguity discussed in Section~\ref{sec:res_supervision}. With $1\%$ of somatic samples, the two spike-generating conductances are already recovered to within $1\%$ ($0.08\%$ for $\bar g_{Na,s}$, $0.94\%$ for $\bar g_{DR,s}$), although the voltage RMSE is still $6$--$10$~mV at this level: as soon as the anchors force the network onto the spiking branch, the shape of the reconstructed spikes constrains $\bar g_{Na,s}$ and $\bar g_{DR,s}$ tightly, even if the timing of individual spikes is not yet exact. The slow conductances remain poorly determined at $1\%$ ($79\%$ error for $\bar g_{M,s}$, $68\%$ for $\bar g_{Ca,s}$) and improve only when the fraction is raised to $5\%$, where their errors fall to $18\%$ and $14\%$. The remainder of this section discusses the $5\%$ estimates in detail.

The clean-data estimates at $5\%$ supervision are reported in the last two columns of Table~\ref{tab:somatic_supervision_ablation}, which lists the true value, the estimated value, and the relative error for each conductance. The sodium and delayed-rectifier potassium conductances are recovered with very high accuracy. The estimate of $\bar g_{Na,s}$ is $119.93~\mathrm{mS/cm^2}$ for a true value of $120~\mathrm{mS/cm^2}$, corresponding to a relative error of $0.06\%$. The estimate of $\bar g_{DR,s}$ is $11.99~\mathrm{mS/cm^2}$ for a true value of $12~\mathrm{mS/cm^2}$, corresponding to a relative error of $0.08\%$.

These near-exact estimates are consistent with the voltage reconstruction in Figure~\ref{fig:pinn-clean} and the low RMSE values in Figure~\ref{fig:vs_abs_error}. Biophysically, this is expected because sodium and delayed-rectifier potassium currents dominate the main spike waveform: the sodium current controls the rapid depolarizing upstroke, whereas the delayed-rectifier potassium current contributes strongly to repolarization. This is the central mechanism of Hodgkin--Huxley spike generation~\cite{hodgkin1952quantitative}. Therefore, any model that accurately reconstructs spike timing and spike shape receives strong information about $\bar g_{Na,s}$ and $\bar g_{DR,s}$.

By contrast, Table~\ref{tab:somatic_supervision_ablation} shows that the M-type potassium and calcium conductances are less tightly identified. The M-type conductance is estimated as $\bar g_{M,s}=0.206~\mathrm{mS/cm^2}$ instead of $0.25~\mathrm{mS/cm^2}$, giving a relative error of $17.77\%$. The calcium conductance is estimated as $\bar g_{Ca,s}=0.114~\mathrm{mS/cm^2}$ instead of $0.10~\mathrm{mS/cm^2}$, giving a relative error of $14.11\%$. This does not indicate a failure of the reconstruction. Rather, it reveals an identifiability hierarchy. The fast sodium and delayed-rectifier potassium currents leave a direct and repeated signature on every spike, whereas the M-type and calcium currents contribute more indirectly through slower features such as adaptation, recovery, and subthreshold depolarization. These slower effects can be partially compensated by other variables or parameters, which is a common difficulty in conductance-based inverse problems~\cite{huys2006efficient,daly2015hodgkin}. Disparate sets of maximal conductances are indeed known to produce nearly indistinguishable
activity in conductance-based neuronal models~\cite{prinz2004similar}, which is the mechanism
underlying such compensations. In the terminology of identifiability analysis, $\bar
g_{\mathrm{M},s}$ and $\bar g_{\mathrm{Ca},s}$ are not structurally unidentifiable here, since
they enter the somatic current balance through distinct terms; they are practically
unidentifiable, in the sense that the amount and the quality of the available observations do
not constrain them tightly~\cite{raue2009identifiability}.

The parameters of the model explain why these two conductances are harder to identify than the fast ones, independently of the estimator used. First, they are small: $\bar g_{M,s}=0.25$ and $\bar g_{Ca,s}=0.1~\mathrm{mS/cm^2}$ are between roughly fifty and a thousand times smaller than $\bar g_{DR,s}$ and $\bar g_{Na,s}$, and the corresponding currents are of the order of the leak current or below. Second, they are slow: the M-type gate $p$ has a time constant of a few tens of milliseconds, and one of the two calcium gates has a time constant of $569$~ms (Table~\ref{tab:kinetics}), longer than the $500$~ms window itself, so that this gate barely departs from its initial value during a run. The signature of these currents in the somatic voltage is therefore a small, slowly varying contribution to the current balance of Eq.~\eqref{eq:voltage_dynamics}, which a few sparse somatic samples and a strongly filtered dendritic trace can only weakly constrain, and which can be partially absorbed by neighbouring terms of the same equation. The sensitivity of the loss to $\bar g_{M,s}$ and $\bar g_{Ca,s}$ is correspondingly low: the loss landscape is nearly flat along these two directions, which produces both the residual bias observed here and the seed-to-seed variability reported in Section~\ref{sec:res_robustness}.

The sparse somatic supervision is important for interpreting these conductance results. Since only a small fraction of somatic samples is included in the loss, the network receives limited direct information about the hidden somatic state. These sparse samples help stabilize the reconstruction and likely improve the partial identifiability of slower conductances compared with a strictly dendrite-only loss. However, they do not fully remove the ambiguity associated with $\bar g_{M,s}$ and $\bar g_{Ca,s}$. Table~\ref{tab:somatic_supervision_ablation} therefore shows both the strength and the limitation of the approach: the dominant spike-generating conductances are strongly identifiable, whereas slower and calcium-related conductances require richer stimulation protocols, longer recordings, or additional measurements. Since $5\%$ is the level at which both the trajectory and the conductances are recovered, this configuration is used, without any other change, in the comparison with the UKF and in the robustness study that follow.

\subsection{The PINN and the Unscented Kalman Filter (UKF)}
\label{sec:ukf}
Before comparing the two reconstructions, it is worth recalling what the two
estimators share. The filter is not an alternative model: its sigma points are
propagated through the same two-compartment vector field used in the physics
residual of Section~\ref{sec:inverse_problem}, with the same gating kinetics, the
same passive parameters and the same coupling conductance
$g_c = 0.1$\,mS/cm$^2$. It assimilates the same dendritic observations and the same $5\%$ somatic samples, over the
same $500$\,ms window, and it estimates the same four somatic maximal
conductances by augmenting the state vector. Any difference between the two
reconstructions therefore comes from the way the model constraint is imposed,
not from the biophysics itself.

Figure~\ref{fig:ukf_traces} shows the reconstructed somatic voltage over the full
window. Both methods detect every spike of the reference simulation, in all four
protocols, and neither produces spurious events, so the comparison is not about
spike detection but about the shape of the residual error. The UKF estimate is
smooth but shifted: during the depolarized plateau of the constant protocol it
remains about $4$--$5$\,mV below the reference, and the same negative offset
reappears in the interspike intervals of the sinusoidal and pulse-train
protocols. The PINN trajectory fluctuates around the reference with a small
ripple inherited from the Fourier features, but without any persistent offset.
The PINN error is thus mainly a variance term, whereas the UKF error contains a
systematic bias.

Figure~\ref{fig:ukf_zoom} shows that this bias also reaches the spike waveform.
The filtered estimate rises slightly too early, repolarizes faster than the
reference and overshoots the afterhyperpolarization by about $6$--$7$\,mV before
settling at a more hyperpolarized level, so the spike is visibly narrower. The
PINN reproduces the spike width and the repolarization phase, and its only
visible defects are a small underestimation of the peak amplitude and the
subthreshold ripple already mentioned.
\begin{figure}[htbp]
\centering
\includegraphics[width=0.7\textwidth]{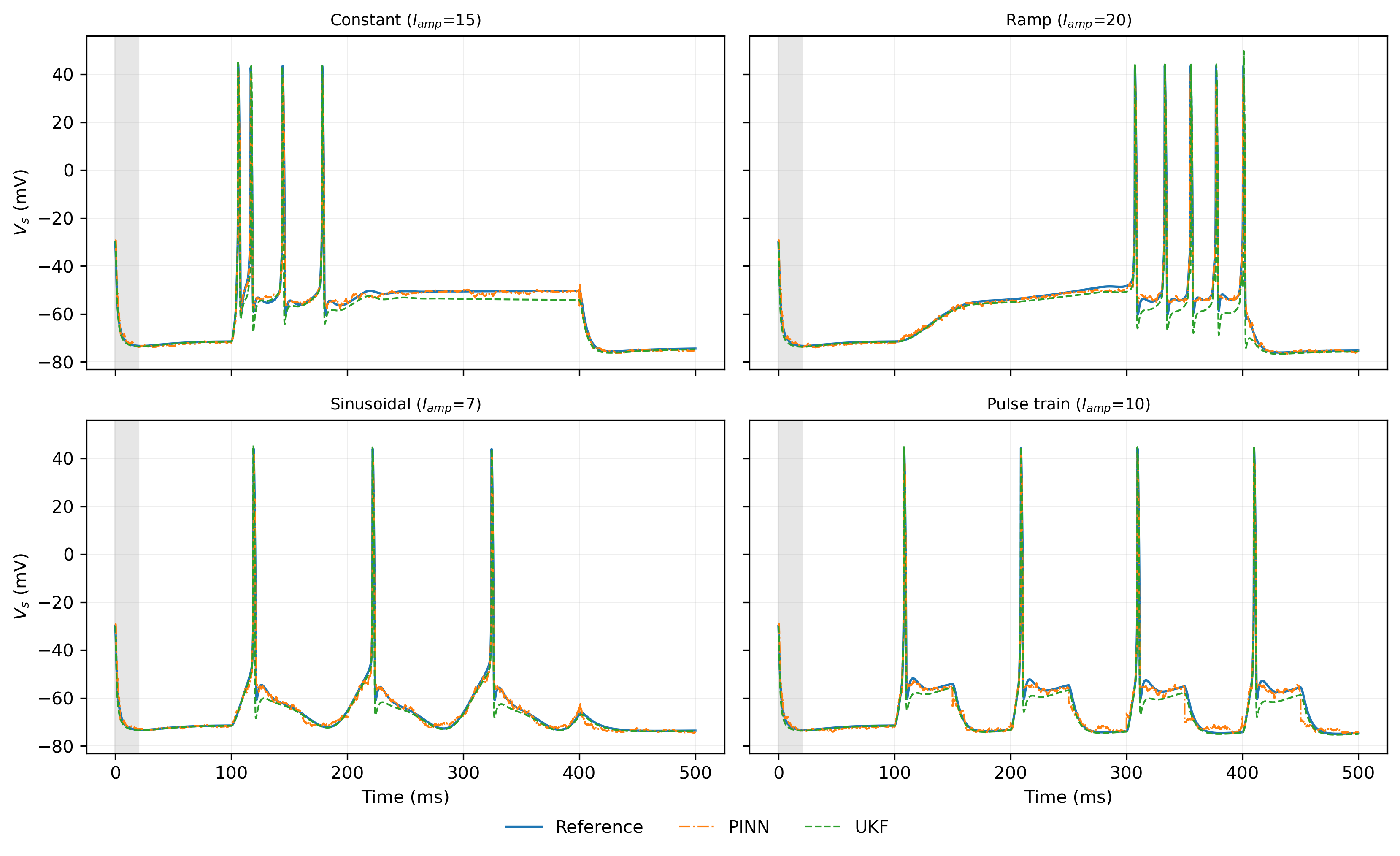}
\caption{Comparison of the PINN and UKF reconstructions of the somatic membrane
potential $V_s(t)$ with 5\% somatic supervision, for the four dendritic
stimulation protocols: reference simulation (solid blue), PINN (orange,
dash-dotted) and augmented-state UKF (green, dashed).}
\label{fig:ukf_traces}
\end{figure}

\begin{figure}[htbp]
\centering
\includegraphics[width=0.7\textwidth]{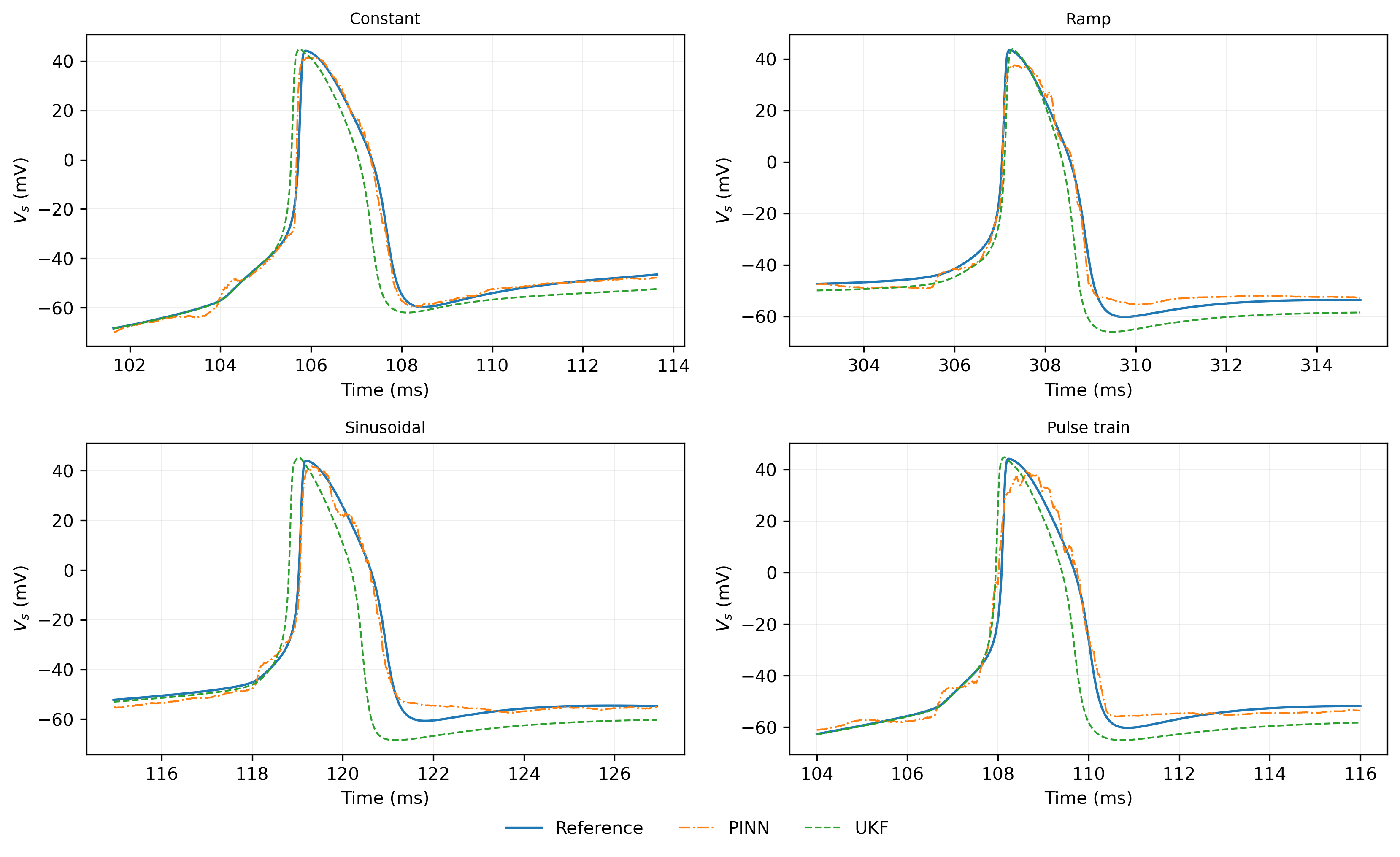}
\caption{Enlargement of Figure~\ref{fig:ukf_traces} on a single action potential,
one per stimulation protocol; same colour code. The UKF spike starts earlier,
repolarises faster and goes deeper into the afterhyperpolarisation, so it appears
narrower than the reference.  The PINN keeps the correct spike width and
repolarisation time course; it only underestimates the peak slightly and shows a
small ripple in the subthreshold parts.}
\label{fig:ukf_zoom}
\end{figure}

\begin{table}[htbp]
\centering
\caption{UKF reconstruction and parameter-estimation results using 5\% somatic data.}
\label{tab:ukf_results_5pct} 
\begin{subtable}[t]{0.52\textwidth}
\centering
\caption{Somatic reconstruction errors obtained with the UKF using 5\% somatic data.}
\label{tab:ukf_reconstruction_5pct}
\small
\setlength{\tabcolsep}{5pt}
\begin{tabular}{lrrr}
\toprule
Protocol & MAE & RMSE & $E_{\max}$ \\
\midrule
Constant   & 2.29 & 5.09 & 83.99 \\
Ramp       & 2.10 & 5.10 & 89.31 \\
Sinusoidal & 1.13 & 3.29 & 64.49 \\
Pulse train & 1.83 & 3.37 & 51.32 \\
\bottomrule
\end{tabular}
\end{subtable}
\hfill
\begin{subtable}[t]{0.45\textwidth}
\centering
\caption{Somatic maximal conductances estimated by the UKF using 5\% somatic data, in mS/cm$^2$.}
\label{tab:cond_ukf_5pct}
\small
\setlength{\tabcolsep}{4pt}
\begin{tabular}{lrrr}
\toprule
Parameter & Truth & Est. & Rel. err. (\%) \\
\midrule
$\bar g_{\mathrm{Na},s}$ & 120  & 161.61 & 34.67  \\
$\bar g_{\mathrm{DR},s}$ & 12   & 17.57  & 46.43  \\
$\bar g_{\mathrm{M},s}$  & 0.25 & 0.574  & 129.75 \\
$\bar g_{\mathrm{Ca},s}$ & 0.10 & 0.0224 & 77.63  \\
\bottomrule
\end{tabular}
\end{subtable}
\end{table}

Table~\ref{tab:ukf_results_5pct} summarizes the UKF results obtained using 5\%
somatic data. The reconstruction error depends on the stimulation protocol. The
lowest RMSE is obtained for the sinusoidal input ($3.29$\,mV), followed by the
pulse train ($3.37$\,mV), whereas the constant and ramp protocols yield RMSE
values of $5.09$ and $5.10$\,mV, respectively. The maximum absolute error ranges
from $51.32$ to $89.31$\,mV, showing that the largest discrepancies occur near
the sharpest voltage variations.

The conductance estimates reported in Table~\ref{tab:cond_ukf_5pct} show relative
errors of $34.67\%$ for $\bar g_{\mathrm{Na},s}$, $46.43\%$ for
$\bar g_{\mathrm{DR},s}$, $129.75\%$ for $\bar g_{\mathrm{M},s}$, and $77.63\%$
for $\bar g_{\mathrm{Ca},s}$. Thus, under this observation setting, the UKF
provides an approximate reconstruction of the somatic trajectory, while the
selected maximal conductances remain only weakly identified.

Placing the two estimators side by side under the same $5\%$ observation setting gives the following picture. For the somatic trajectory, the RMSE of the PINN (Table~\ref{tab:ablation_sparse_compact}, $5\%$ column) is lower than that of the UKF (Table~\ref{tab:ukf_reconstruction_5pct}) in every protocol: by a factor of $2.8$ for the constant input ($1.79$ against $5.09$~mV), $2.2$ for the ramp ($2.37$ against $5.10$~mV), $1.5$ for the sinusoidal input ($2.23$ against $3.29$~mV) and $1.4$ for the pulse train ($2.37$ against $3.37$~mV). The MAE is likewise smaller for the PINN in all four protocols, although only marginally for the
sinusoidal input ($1.10$ against $1.13$~mV), and the maximum error of the PINN is smaller in
three protocols out of four ($34.9$--$54.0$~mV against $64.5$--$89.3$~mV), the pulse train being
the only case where the two methods are equivalent on this metric ($52.5$ against $51.3$~mV). The gap is wider for the conductances: the PINN recovers $\bar g_{Na,s}$ and $\bar g_{DR,s}$ to within $0.1\%$ where the UKF errs by $35\%$ and $46\%$, and its errors on $\bar g_{M,s}$ and $\bar g_{Ca,s}$ ($18\%$ and $14\%$) are five to seven times smaller than those of the filter ($130\%$ and $78\%$).

The two results are related. The UKF estimates $\bar g_{Na,s}$, $\bar g_{DR,s}$ and $\bar g_{M,s}$ above their true values and $\bar g_{Ca,s}$ below it, i.e.\ it describes a soma with a stronger fast inward current, a stronger repolarizing current, a stronger slow potassium current and a weaker slow depolarizing current. This is precisely the waveform seen in Figures~\ref{fig:ukf_traces} and~\ref{fig:ukf_zoom}: an earlier rise, a faster repolarization, a deeper afterhyperpolarization and an interspike level several millivolts too low. The systematic voltage offset of the filter and its biased conductances are two expressions of the same estimate. The origin of this bias lies in the sequential nature of the filter in a weakly observed problem. Between two somatic samples, which are separated on average by twenty time steps at $5\%$, the filter can only rely on the dendritic observation, which carries little information about $V_s$ under weak coupling, and on the model prediction; the augmented parameters, treated as constant states, are then adjusted locally to reduce the innovations, and the Gaussian approximation on which the unscented transform rests is stretched by the strongly nonlinear spike upstroke. Each somatic sample corrects the state but only partially the parameters, and the estimate settles on a biased set that is locally consistent with the data. The PINN, in contrast, imposes the same equations at all collocation points simultaneously, so that every spike of every protocol contributes to the same four global parameters and the reconstructed trajectory has to be consistent with the whole window at once. A fixed-interval smoother would recover part of this global information and could reduce the offset; the comparison reported here is with the filter, which is the standard sequential estimator, and it isolates the effect of imposing the model globally rather than sequentially.

\subsection{Robustness to noise and random initialization}
\label{sec:res_robustness}

The two experiments of this section use the $5\%$ configuration of Section~\ref{sec:res_supervision}, with the same architecture, loss weights, optimizer and somatic sampling mask; only the observation noise or the random seed is changed. The robustness of the method to measurement noise is summarized in Table~\ref{tab:noise_robustness}. Gaussian noise with standard deviation $\sigma=0.5$, $1$, and $2$ mV was added to the voltage observations, in addition to the clean case $\sigma=0$. As expected, the somatic reconstruction error increases with the noise level: $\mathrm{RMSE}(V_s)$ rises from $1.6902$ mV in the clean case to $1.8699$ mV for $\sigma=0.5$ mV, $1.9957$ mV for $\sigma=1$ mV, and $2.5936$ mV for $\sigma=2$ mV. This increase shows that noise degrades the reconstruction quantitatively, but the error remains moderate relative to the amplitude of action potentials.

Two features of this progression are worth noting. First, the total increase of the RMSE between $\sigma=0$ and $\sigma=2$~mV (about $0.9$~mV) is smaller than the noise standard deviation itself, so the inversion does not amplify the measurement noise, which is not guaranteed for a partially observed inverse problem. Second, since the reconstruction plotted in Figure~\ref{fig:noise_sigma2} remains smooth in the subthreshold intervals, the additional error is most likely concentrated near the spikes, where the noisy dendritic trace and the noisy somatic anchors perturb the timing of the reconstructed events by a fraction of a millisecond, with the large instantaneous errors discussed in Section~\ref{sec:res_supervision}.

The conductance estimates in Table~\ref{tab:noise_robustness} show that the spike-generating conductances remain highly robust under noisy observations. Across all noise levels, the relative error for $\bar g_{Na,s}$ remains between $0.06\%$ and $0.11\%$, while the error for $\bar g_{DR,s}$ remains between $0.02\%$ and $0.16\%$. Even for $\sigma=2$ mV, both conductances are recovered with errors below $0.2\%$. This confirms that the dominant fast conductances are strongly constrained by the spike dynamics and by the physics-informed residuals.

The slower conductances show a different behavior in Table~\ref{tab:noise_robustness}. The M-type conductance error remains around $17.6$--$17.8\%$ for $\sigma=0$, $0.5$, and $1$ mV, then decreases to $8.25\%$ for $\sigma=2$ mV. This non-monotonic behavior suggests that the M-type estimate is affected not only by noise amplitude, but also by optimization variability and parameter compensation. The calcium conductance error increases from $14.11\%$ in the clean case to $25.12\%$ at $\sigma=1$ mV and remains elevated at $21.97\%$ for $\sigma=2$ mV. Thus, the noisy-data results confirm that calcium-related mechanisms are less robustly identifiable than the dominant spike-generating currents. Read together with the seed experiment below, these fluctuations point to a single explanation: along the $\bar g_{M,s}$ and $\bar g_{Ca,s}$ directions the loss is nearly flat, so that any perturbation of the training problem, whether measurement noise or a different initialization, displaces the estimate along this flat direction without a systematic trend. The lower M-type error at $\sigma=2$~mV should therefore not be read as an improvement brought by the noise, but as one more sample of this variability.
\begin{figure}[h!]
\centering
\includegraphics[width=0.8\textwidth]{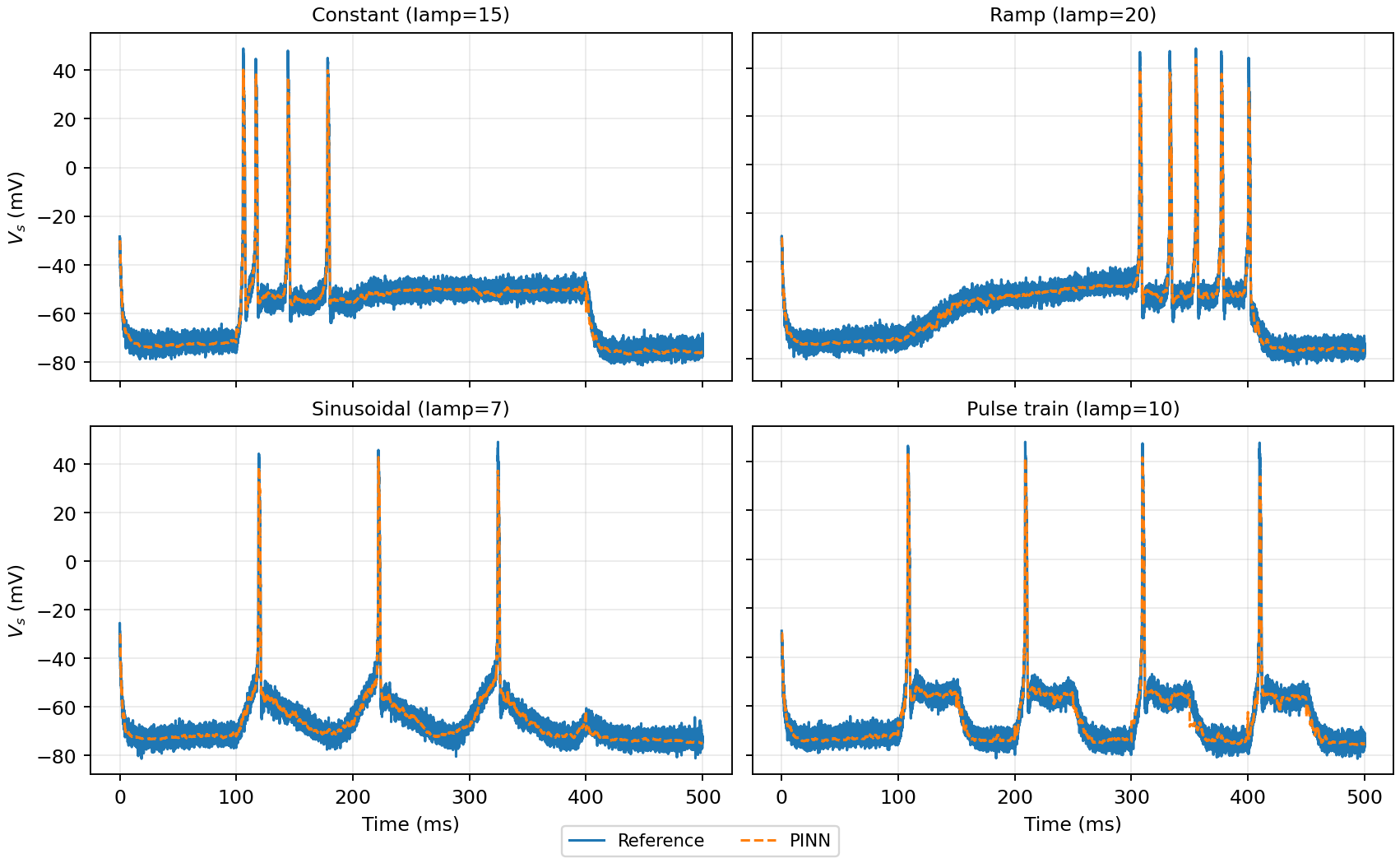}
\caption{
Somatic voltage reconstruction under strong Gaussian measurement noise
\((\sigma=2~\mathrm{mV})\). The PINN reconstruction is compared with the
reference somatic voltage for the four dendritic stimulation protocols:
constant step, ramp, rectified sinusoidal input, and pulse train. Despite the
noisy observations, the PINN preserves the main subthreshold dynamics and
spike timing, illustrating the regularizing effect of the physics-informed
constraints.
}
\label{fig:noise_sigma2}
\end{figure}

\begin{table}[h]
\centering
\caption{Robustness of the PINN to Gaussian voltage noise.
The noise level \(\sigma\) is the standard deviation of the additive Gaussian
noise applied to the voltage observations. RMSE values are reported in mV.
Relative errors are computed with respect to the true somatic conductances.}
\label{tab:noise_robustness}
\begin{tabular}{ccccccc}
\toprule
\(\sigma\)

& RMSE\((V_s)\)
& Err. \(\bar g_{\mathrm{Na},s}\)
& Err. \(\bar g_{\mathrm{DR},s}\)
& Err. \(\bar g_{\mathrm{M},s}\)
& Err. \(\bar g_{\mathrm{Ca},s}\) \\
(mV)  & (mV) & (\%) & (\%) & (\%) & (\%) \\
\midrule
0  & 1.6902 & 0.06 & 0.08 & 17.77 & 14.11 \\
0.5  & 1.8699 & 0.06 & 0.02 & 17.62 & 15.79 \\
1 & 1.9957 & 0.06 & 0.10 & 17.77 & 25.12 \\
2 & 2.5936 & 0.11 & 0.16 & 8.25  & 21.97 \\
\bottomrule
\end{tabular}
\end{table}
Figure~\ref{fig:noise_sigma2} provides a visual example of the reconstruction under strong noise, with $\sigma=2$ mV. The reference voltage contains visible noisy fluctuations, whereas the PINN reconstruction remains smooth and preserves the main spike timing, subthreshold structure, and recovery phases. This illustrates the regularizing role of the physics-informed loss: instead of fitting high-frequency measurement noise, the network favors trajectories that remain consistent with the two-compartment Hodgkin--Huxley dynamics. The sparse somatic samples also help anchor the hidden voltage under noisy conditions, but because these samples are themselves sparse and noisy, they cannot fully resolve the identifiability limitations of the slower conductances.

Finally, Table~\ref{tab:robustness_clean_8runs} evaluates whether the clean-data results depend on a favorable random initialization. The same clean experiment was repeated over eight independent random seeds, while keeping the dataset, stimulation protocols, architecture, hyperparameters, optimizer, and sparse somatic sampling mask fixed. The somatic reconstruction error remains stable across runs, with $\mathrm{RMSE}(V_s)=1.83\pm0.19$ mV and a min--max range of $1.42$--$2.03$ mV. The best training loss is $12.09\pm1.28$, indicating moderate optimization variability, which is expected in PINN training because the loss combines data fidelity, sparse somatic supervision, and stiff ODE residuals near spikes.

The conductance estimates in Table~\ref{tab:robustness_clean_8runs} confirm the same identifiability hierarchy observed in Tables~\ref{tab:somatic_supervision_ablation} and~\ref{tab:noise_robustness}. The sodium conductance is highly reproducible, with $\bar g_{Na,s}=119.91\pm0.04~\mathrm{mS/cm^2}$, a coefficient of variation of $0.03\%$, and a relative error of $0.08\%$. The delayed-rectifier potassium conductance is also stable, with $\bar g_{DR,s}=11.97\pm0.10~\mathrm{mS/cm^2}$, a coefficient of variation of $0.83\%$, and a relative error of $0.25\%$. In contrast, the M-type conductance has a larger coefficient of variation of $20.24\%$, and the calcium conductance has a coefficient of variation of $9.57\%$. This shows that the main limitation is not a global instability of the PINN optimization. Rather, the larger variability reflects the weaker identifiability of slow and calcium-related conductances under the present stimulation protocols and sparse somatic guidance. The min--max ranges of Table~\ref{tab:robustness_clean_8runs} add one observation: across the eight runs $\bar g_{M,s}$ remains at or below its true value ($0.116$--$0.251~\mathrm{mS/cm^2}$, mean $0.211$), whereas $\bar g_{Ca,s}$ lies mostly above it ($0.093$--$0.129~\mathrm{mS/cm^2}$, mean $0.114$), and the same direction is found in the noise experiment (Table~\ref{tab:noise_robustness}). The residual error on these two conductances is thus a small systematic bias superimposed on a large variance. Such a bias is consistent with a partial compensation between the slow somatic currents and the other terms of the current balance, and not with an isolated optimization failure, which would produce errors of random sign. The reconstruction error itself, with a coefficient of variation of $10\%$ and a range of $1.42$--$2.03$~mV, shows that the somatic trajectory is recovered with the same quality by every run, independently of where the slow conductances settle: the trajectory is well determined even where these two parameters are not.

The results of the eight independent training runs are summarized in Table~\ref{tab:robustness_clean_8runs}.
For a quantity $z$ measured across the eight runs, we report the empirical mean and standard deviation as
\[
\overline z
=
\frac{1}{8}
\sum_{j=1}^{8} z_j,
\qquad
\mathrm{Std}(z)
=
\sqrt{
\frac{1}{7}
\sum_{j=1}^{8}
\left(z_j-\overline z\right)^2
}.
\]
The coefficient of variation is defined as
\[
\mathrm{CV}
=
100
\frac{\mathrm{Std}(z)}{|\overline z|}.
\]

\begin{table}[htbp]
\centering
\caption{Training reproducibility across eight independent random initializations in the clean-data case.
All runs used the same dataset, network architecture, loss weights, optimization settings, and sparse somatic sampling mask; only the random seed was changed.
The table reports full-dataset evaluation metrics for voltage reconstruction, optimization, and conductance estimation.}

\label{tab:robustness_clean_8runs}
\small
\begin{tabular}{lccccc}
\toprule
Quantity & Reference value & Mean $\pm$ Std & Min--Max & CV (\%) & Relative error (\%) \\
\midrule
\multicolumn{6}{l}{\textit{Voltage reconstruction}} \\
RMSE $V_s$  & -- & $1.83 \pm 0.19$ & $1.42$--$2.03$ & $10.30$ & -- \\
\midrule
\multicolumn{6}{l}{\textit{Optimization}} \\
Best training loss & -- & $12.09 \pm 1.28$ & $10.13$--$13.81$ & $10.59$ & -- \\
\midrule
\multicolumn{6}{l}{\textit{Conductance estimation}} \\
$\bar g_{\mathrm{Na},s}$  & $120$ & $119.91 \pm 0.04$ & $119.84$--$119.96$ & $0.03$ & $0.08$ \\
$\bar g_{\mathrm{DR},s}$   & $12$  & $11.97 \pm 0.10$  & $11.79$--$12.09$  & $0.83$ & $0.25$ \\
$\bar g_{\mathrm{M},s}$   & $0.25$   & $0.211 \pm 0.043$ & $0.116$--$0.251$ & $20.24$ & $15.77$ \\
$\bar g_{\mathrm{Ca},s}$  & $0.10$   & $0.114 \pm 0.011$ & $0.093$--$0.129$ & $9.57$ & $13.79$ \\
\bottomrule
\end{tabular}
\end{table}

Taken together, Figures~\ref{fig:pinn-clean}--\ref{fig:noise_sigma2} and Tables~\ref{tab:ablation_sparse_compact}--\ref{tab:robustness_clean_8runs} show that the proposed PINN can reconstruct the hidden somatic output $V_s(t)$ from dense dendritic observations complemented by sparse somatic anchors in the weak-coupling regime. The method accurately recovers the main spike waveform, remains stable under Gaussian voltage noise, and is reproducible across independent random initializations. The dominant spike-generating conductances $\bar g_{Na,s}$ and $\bar g_{DR,s}$ are strongly identifiable, whereas $\bar g_{M,s}$ and $\bar g_{Ca,s}$ remain partially identifiable because they influence slower and less directly observed features of the trajectory. These results motivate two natural extensions: first, designing richer stimulation protocols and additional sparse measurements to better constrain slow conductances; second, validating the approach on real electrophysiological recordings, where dendritic or somatic measurements are sparse and noisy. This experimental validation will be an essential step toward assessing the practical usefulness of the method beyond synthetic data.

\section{Conclusion}

This work addressed the reconstruction of the somatic membrane potential of a weakly coupled two-compartment Hodgkin--Huxley neuron ($g_c=0.1~\mathrm{mS/cm^2}$) from dense dendritic recordings, the known injected current and a small number of somatic samples, using a PINN that also estimates four somatic maximal conductances. Three conclusions can be drawn from the synthetic experiments.

First, the somatic-supervision ablation clarifies the observation requirements of the proposed
approach. Under the present training configuration, the PINN did not recover the somatic
trajectory when trained without somatic observations: the reconstruction stays close to a
smooth subthreshold level, all spikes are missed (RMSE of $10$--$14$~mV) and the conductances
converge to a low-excitability set far from the reference values. Introducing $1\%$ of somatic
samples produced a substantial improvement, in that every spike is then present and the two
spike-generating conductances are already recovered to within $1\%$; the RMSE nevertheless
remains at $6$--$10$~mV because of amplitude and timing errors on individual spikes. Raising
the sampling fraction to $5\%$ yielded close agreement with the reference trajectories, with an
RMSE of about $2$~mV including the spikes and errors below $0.1\%$ on $\bar g_{\mathrm{Na},s}$
and $\bar g_{\mathrm{DR},s}$. The successful reconstructions reported here therefore depend on
sparse somatic anchoring in addition to dense dendritic observations, known stimulation
currents, and the imposed two-compartment equations.

Second, with exactly the same model and the same $5\%$ observations, the PINN is more accurate than the sequential UKF baseline. Its RMSE is lower in every protocol, it shows no systematic offset, whereas the filter settles $4$--$5$~mV below the reference between spikes, and the fast conductances are recovered with errors below $0.1\%$ instead of $35$--$46\%$. Since the two estimators share the equations, the parameters and the data, this difference is attributable to the way the constraint is imposed, globally over the observation window rather than sequentially. The PINN reconstruction also degrades gracefully under Gaussian observation noise up to $2$~mV, with an RMSE increase smaller than the noise level itself, and it is reproducible across eight random initializations ($\mathrm{RMSE}=1.83\pm0.19$~mV).

Third, the residual weakness of the method concerns the two small and slow somatic conductances,
$\bar g_{\mathrm{M},s}$ and $\bar g_{\mathrm{Ca},s}$, which are recovered with errors of
$14$--$18\%$ and vary from seed to seed. Several observations indicate that this limitation is a
property of the inverse problem rather than of the PINN. These currents are two to three orders
of magnitude smaller than the spike-generating ones and act on time scales, tens of milliseconds
for the M-type gate and $569$~ms for one of the calcium gates, comparable to or longer than the
$500$~ms window, so that their signature in a filtered dendritic trace and in a few somatic
samples is faint and can be partly compensated by other terms of the current balance. The UKF,
constrained by the same equations and data, identifies them even less well ($78$--$130\%$
error), and the seed-to-seed variability is confined to these two parameters while
$\bar g_{\mathrm{Na},s}$ and $\bar g_{\mathrm{DR},s}$ are reproducible to within
$0.03$--$0.8\%$. This is the situation of weak practical identifiability discussed in
Section~\ref{sec:res_conductances}: the difficulty lies in what the available observations can
say about these two parameters, and no estimator can remove it without more informative data
\cite{huys2006efficient,daly2015hodgkin,prinz2004similar,raue2009identifiability}. Importantly,
the somatic trajectory itself is recovered with the same quality in every run regardless of
where these two parameters settle, so the main output of the method, the hidden somatic voltage,
is not affected by this limitation.

These conclusions are restricted to synthetic data, four stimulation protocols, and \(g_c=0.1\ \mathrm{mS/cm^2}\). Whether dendrite-only reconstruction can be achieved using alternative formulations, stronger constraints, or more informative stimulation protocols remains an open question; likewise, protocols designed to excite the slow currents, recordings longer than their time constants, or an additional slow observable such as intracellular calcium are the natural means of improving the identification of $\bar g_{M,s}$ and $\bar g_{Ca,s}$.

\section{Perspectives}
\label{sec:perspectives}

\paragraph{Generalization to $N$-compartment models.}
The proposed PINN formulation is not intrinsically tied to a two-compartment topology. 
It can be extended, in principle, to neurons with a larger number of electrically coupled compartments and more complex dendritic morphologies. 
In such models, information propagates through the dendritic tree according to membrane current-balance equations, axial coupling terms, and possibly cable-type dynamics. 
A PINN is a natural candidate in this setting because the physics residuals can be imposed
directly on the governing equations, while data-fidelity terms are applied only at the recorded compartments. 
Thus, sparse measurements from a subset of compartments can constrain the reconstruction of unobserved voltages elsewhere in the morphology, as illustrated in Figure~\ref{fig:generalization_N_compartments}. 
The main challenge is identifiability: as the number of compartments and unknown conductances increases, additional measurements, parameter sharing, or regularization strategies may be required.

\begin{figure}[h!]\centering\includegraphics[width=0.7\linewidth]{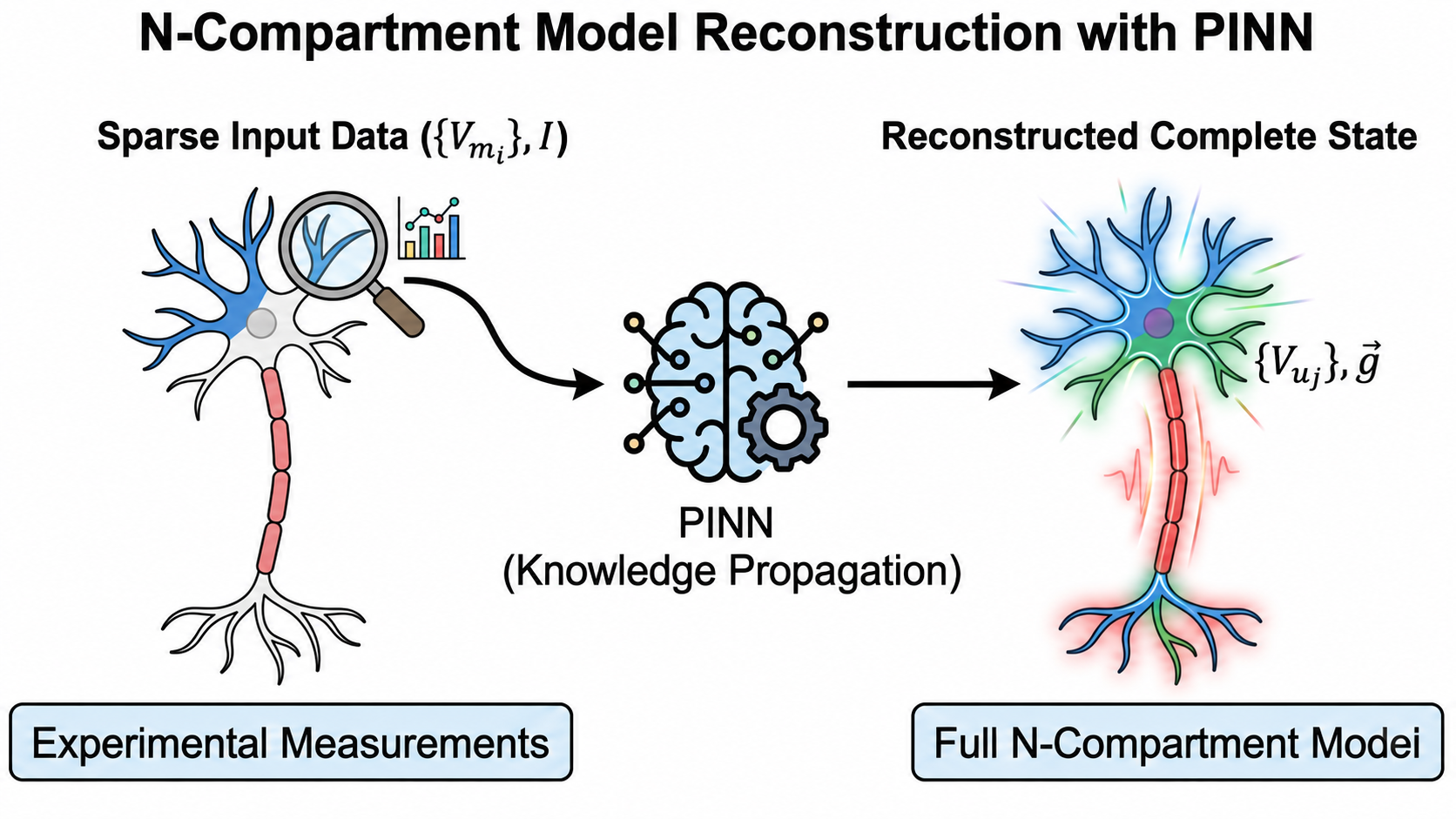}
\caption{\textbf{Generalization to $N$-compartment reconstruction with physics-informed learning.}
Conceptual extension of the proposed PINN framework to a neuron with multiple electrically coupled compartments. 
\textit{Left:} only a small subset of compartmental voltages is experimentally observed, together with the applied stimulation protocol. 
\textit{Center:} the PINN propagates the available information through the biophysical equations, including membrane current balance, gating kinetics, and axial coupling. 
\textit{Right:} the method reconstructs the voltage dynamics of unobserved compartments and estimates selected biophysical parameters, such as maximal conductances and coupling strengths, yielding a physically consistent full multi-compartment model.}

\label{fig:generalization_N_compartments}
\end{figure}
\paragraph{From single-neuron inference to network synchronization.}
A second extension is to move from a single multi-compartment 
neuron to networks of coupled neurons, where synchrony emerges 
from the interplay of intrinsic dynamics, connectivity, and 
external drive. In this context, the PINN can be used not only 
for state and parameter identification, but also for inverse 
design: determining which stimulation frequencies promote or suppress synchronization (Figure~\ref{fig:pinn_critical_frequency}). This 
perspective is particularly relevant to understanding transitions 
to pathological rhythmicity in conditions such as epilepsy or 
Parkinson's disease.

\begin{figure}[h!]
    \centering
    \includegraphics[width=0.82\linewidth]{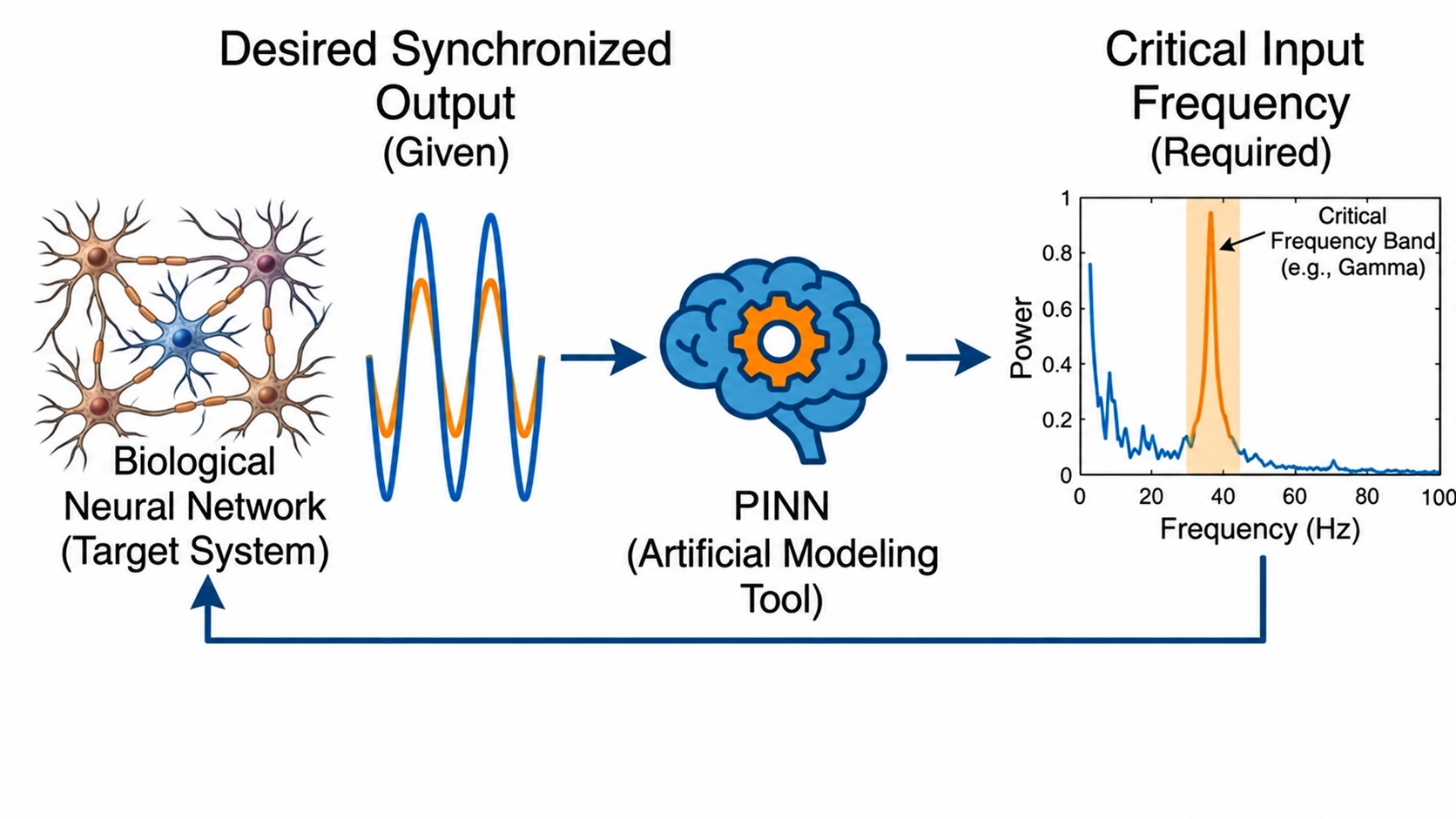}
    \caption{\textbf{PINN-based inference of critical stimulation frequency for synchronization.}
    \emph{Left to center:} a desired synchronized network output (e.g., coherent oscillations) is specified for a biological neural network.
    \emph{Right:} the PINN formulates an inverse problem where the input frequency is optimized under network dynamics constraints; the power spectral density (PSD) highlights a critical frequency band (illustrated here in the gamma range) where oscillatory power and synchrony are maximized.}
    \label{fig:pinn_critical_frequency}
\end{figure}

\newpage
\section*{Availability of data and material}
All data were generated by numerical simulations of the ODE model implemented in MATLAB. The simulation scripts, parameter sets, and the generated output data supporting this study are provided in the Supplementary Materials.

\section*{Competing interests}
The authors declare that they have no competing interests.

\section*{Funding}
This research received no external funding.

\section*{Authors  contributions}
Conceptualization: A.O.,  B.A. and M.A.A.-A.
Methodology: A.O. , B.A. and M.A.A.-A.
Software: A.O.
Validation: A.O, B.A. and M.A.A.-A.
Investigation: A.O.
Writing-original draft: A.O., B.A. and M.A.A.-A.
Writing—review \& editing: A.O., B.A. and M.A.A.-A.
Visualization: A.O.
Supervision: B.A. and M.A.A.-A.
All authors read and approved the final manuscript.

\section*{Acknowledgements}
The authors gratefully acknowledge the PhD Summer School in Physics-Informed Neural Networks and Applications (PINNs 2025), held in Stockholm, Sweden (15-30 June 2025). This work was developed as part of the contributions and follow-up activities proposed after the summer school, and benefited significantly from the lectures and discussions.


\end{document}